\documentclass{article}
\usepackage{mathptmx}       % selects Times Roman as basic font
\usepackage{graphicx}        % standard LaTeX graphics tool
\usepackage{multicol}        % used for the two-column index
\usepackage[bottom]{footmisc}% places footnotes at page bottom
\usepackage{tensor, cite} 

\usepackage{amsmath,amssymb,amsfonts,amsthm,amscd,latexsym,amssymb}%,bbm,phystex} 
\usepackage[mathscr]{eucal}
\usepackage{mathrsfs}

\DeclareMathAlphabet{\mathcal}{OMS}{cmsy}{m}{n}

\makeatletter
\g@addto@macro\bfseries{\boldmath}
\makeatother

\DeclareMathAlphabet\BEuScript{U}{eus}{b}{n}

\def\SLL{\mathrm{SL}} % Gruppo speciale lineare
\def\Lcc{L_\mathrm{c}} % Isomorfismo time-slice per cospinori
\newcommand{\CC}{{\mathbb C}}
\newcommand{\NN}{{\mathbb N}}
\newcommand{\RR}{{\mathbb R}}
\newcommand{\bS}{{\mathbb S}}

\newcommand{\Hom}{\mathrm{Hom}}

\usepackage{bm}
\DeclareMathAlphabet{\mathbfsf}{\encodingdefault}{\sfdefault}{bx}{n}

\DeclareBoldMathCommand\Fb{F}
\DeclareBoldMathCommand\Mb{M}
\DeclareBoldMathCommand\Nb{N}
\DeclareBoldMathCommand\Pb{P}
\DeclareBoldMathCommand\Ob{O}
\DeclareBoldMathCommand\rb{R}
\DeclareBoldMathCommand\ab{a}
\DeclareBoldMathCommand\bb{b}
\DeclareBoldMathCommand\cb{c}
\DeclareBoldMathCommand\eb{e}
\DeclareBoldMathCommand\ib{i}
\DeclareBoldMathCommand\jb{j}
\DeclareBoldMathCommand\kb{k}
\DeclareBoldMathCommand\pb{p}
\DeclareBoldMathCommand\rb{r}
\DeclareBoldMathCommand\ub{u}
\DeclareBoldMathCommand\vb{v}
\DeclareBoldMathCommand\xb{x}

\newcommand{\dvol}{{\rm dvol}}

\newcommand{\id}{\text{id}}

\newcommand{\Ac}{{\mathcal{A}}}

\newcommand{\Fc}{{\mathcal{F}}}

\newcommand{\Sol}{\mathsf{ Sol}}

\newcommand{\II}{\leavevmode\hbox{\rm{\small1\kern-3.8pt\normalsize1}}}

\newcommand{\Tr}{\textrm{Tr}\,}
\newcommand{\supp}{\textrm{supp}\,}

\newcommand{\Mc}{{\mathcal M}}

\newcommand{\ogth}{{\mathfrak o}}
\newcommand{\tgth}{{\mathfrak t}}

\newcommand{\nb}{{\boldsymbol{n}}}
\usepackage[english]{babel} % English hyphenation
\usepackage{slashed} % Notazione barrata

\def\func{C^\infty} % Funzioni lisce
\def\cc{\func_0} % Funzioni lisce a supporto compatto
\def\ctc{\func_{tc}} % Funzioni lisce a supporto tc
\def\sect{\Gamma} % Sezioni
\def\sectc{\sect_0} % Sezioni a supporto compatto
\def\sectsc{\sect_{sc}} % Sezioni a supporto spacelike compact
\def\secttc{\sect_{tc}} % Sezioni a supporto timelike compact
\def\f{\Omega} % Forme
\def\fc{\f_0} % Forme a supporto compatto

\def\dd{\mathrm{d}} % Differenziale
\def\de{\delta} % Codifferenziale

\def\even{\mathrm{even}}

\def\lan{\langle} % Appreviazione di \langle
\def\ra{\rangle} % Abbreviazione di \rangle

\def\ClObs{\mathcal{E}} % Osservabili classiche
\def\Solsc{\Sol_{sc}} % Soluzioni sc

\def\Oc{\mathcal{O}} % Regione di uno spaziotempo
\def\Dc{\mathcal{D}} % Algebra di Dirac
\def\MM{\mathbb{M}} % Spazio di Minkowski
\def\ChS{\Gamma} % Simboli di Christoffel

\def\SO{\mathrm{SO}} % Gruppo speciale ortogonale
\def\Spin{\mathrm{Spin}} % Gruppo spin
\def\Matr{\mathrm{M}} % Matrici
\def\ol{\overline} % Coniugazione complessa
\def\Cs{C_\mathrm{s}} % Coniugazione di carica per spinori
\def\Cc{C_\mathrm{c}} % Coniugazione di carica cospinori
\def\nss{\slashed{\nabla}_\mathrm{s}} % Nabla-slash per spinori
\def\nsc{\slashed{\nabla}_\mathrm{c}} % Nabla-slash per cospinori
\def\Ps{P_\mathrm{s}} % Eq. Dirac per spinori
\def\Pc{P_\mathrm{c}} % Eq. Dirac per cospinori
\def\Es{E_\mathrm{s}} % Prop. caus. per spinori
\def\Ec{E_\mathrm{c}} % Prop. caus. per cospinori
\def\Fs{F_\mathrm{s}} % Prop. caus. per \Ps^2
\def\Fc{F_\mathrm{c}} % Prop. caus. per \Pc^2
\def\sSol{\Sol^\mathrm{s}} % Spinori on-shell
\def\cSol{\Sol^\mathrm{c}} % Cospinori on-shell
\def\sSolsc{\Sol^\mathrm{s}_{sc}} % Spinori on-shell spacelike compact
\def\cSolsc{\Sol^\mathrm{c}_{sc}} % Cospinori on-shell spacelike compact
\def\sVan{N^\mathrm{s}} % Funzionali per spinori che si annullano on-shell
\def\cVan{N^\mathrm{c}} % Funzionali per cospinori che si annullano on-shell
\def\sClObs{\ClObs^\mathrm{s}} % Osservabili classiche per spinori
\def\cClObs{\ClObs^\mathrm{c}} % Osservabili classiche per cospinori
\def\sHerm{\mathrm{h}^\mathrm{s}} % Forma hermitiana per osservabili spinori
\def\cHerm{\mathrm{h}^\mathrm{c}} % Forma hermitiana per osservabili cospinori
\def\sHermSol{\mathrm{H}^\mathrm{s}} % Forma hermitiana per soluzioni spinori
\def\cHermSol{\mathrm{H}^\mathrm{c}} % Forma hermitiana per soluzioni cospinori
\def\Ls{L_\mathrm{s}} % Isomorfismo time-slice per spinori
\newcommand{\ips}[2]{(#1,#2)_\mathrm{s}} % Forma hermitiana tra spinori
\newcommand{\ipc}[2]{(#1,#2)_\mathrm{c}} % Forma hermitiana tra cospinori
\newcommand{\ipf}[2]{(#1,#2)} % Pairing tra forme differenziali

\newtheorem{theorem}{Theorem}[section]
\newtheorem{lemma}[theorem]{Lemma}
\newtheorem{proposition}[theorem]{Proposition}

\newtheorem{definition}[theorem]{Definition}
\newtheorem{example}[theorem]{Example}
\newtheorem{rem}[theorem]{Remark}

\newcommand{\Endo}{\mathrm{End}}

\makeindex             % used for the subject index
\begin{document}

\par 
\bigskip 
\LARGE 
\noindent 
{\bf Models of free quantum field theories on curved backgrounds} 
\bigskip 
\par 
\rm 
\normalsize

\large
\noindent {\bf Marco Benini$^{1,2,a}$},
{\bf Claudio Dappiaggi$^{2,3,b}$}, \\
\par
\small
\noindent $^1$ Department of Mathematics, Heriot-Watt University, Colin Maclaurin Building, Riccarton, Edinburgh EH14 4AS, United Kingdom.\smallskip

\noindent $^2$
Dipartimento di Fisica, Universit\`a degli Studi di Pavia,
Via Bassi, 6, I-27100 Pavia, Italy.\smallskip

\noindent $^3$  Istituto Nazionale di Fisica Nucleare - Sezione di Pavia,
Via Bassi, 6, I-27100 Pavia, Italy.\smallskip

\bigskip

\noindent $^a$ mbenini87@gmail.com,
$^b$  claudio.dappiaggi@unipv.it 
 \normalsize

\par 
 
\rm\normalsize 
\bigskip
\noindent {\small Version of \today}

%\linespread{1.5} 
\rm\normalsize 
 
%%%%%%%%%%%% Date %%%%%%%%%%%%%%%%%%%%%%%%%% 
 
\par 
\bigskip 

\noindent 
\small 
{\bf Abstract}. Free quantum field theories on curved backgrounds are discussed via three explicit examples: the real scalar field, the Dirac field and the Proca field. The first step consists of outlining the main properties of globally hyperbolic spacetimes, that is the class of manifolds on which the classical dynamics of all physically relevant free fields can be written in terms of a Cauchy problem. The set of all smooth solutions of the latter encompasses the dynamically allowed configurations which are used to identify via a suitable pairing a collection of classical observables. As a last step we use such collection to construct a $*$-algebra which encodes the information on the dynamics and on the canonical commutation or anti-commutation relations depending whether the underlying field is a Fermion or a Boson.

\section{Geometric data}\label{sec:1}

Goal of this section is to introduce all geometric concepts and tools which are necessary to discuss both the classical dynamics and the quantization of a free quantum field on a curved background. We assume that the reader is familiar with the basic notions of differential geometry and, to a minor extent, of general relativity. Therefore we will only sketch a few concepts and formulas, which we will use throughout this paper; a reader interested to more details should refer to \cite{BGP, BEE, Wald}, yet paying attention to the conventions used here, which differ from time to time from those in the cited references.

Our starting point is $\Mc$, a smooth manifold which is endowed with a (smooth) Lorentzian metric $g$ of signature $(+,-,\dots,-)$. Furthermore, although the standard generalizations to curved backgrounds of the field theories on Minkowski spacetime, on which the current models of particle physics are based, entail that $\Mc$ ought to be four dimensional, in this paper we shall avoid this assumption. The only exception will be Section \ref{subDirac}, where we will describe Dirac spinors in four dimensions only, for the sake of simplicity. We introduce a few auxiliary, notable tensors. We employ an abstract index notation\footnote{Notice that, in this paper, we employ the following convention for the tensor components: Latin indices, $a,b,c,\dots$, are used for abstract tensor indices, Greek ones, $\mu,\nu,\dots$ for coordinates, while $i,j,k$ are used for spatial components or coordinates.} and we stress that our conventions might differ from those of many textbooks, 
{\it e.g.} \cite{Wald}. As a starting point, we introduce the Riemann tensor $\textrm{Riem}:T\Mc\otimes T\Mc\to\Endo(T\Mc)$, defined using the abstract index notation by the formula $(\nabla_a\nabla_b-\nabla_b\nabla_a)v^d=R_{abc}^{\hphantom{abc}d} v^c$, where $v$ is an arbitrary vector field and $\nabla$ is the covariant derivative. The Ricci tensor is instead $\textrm{Ric}:T\Mc^{\otimes 2}\to\RR$ and its components are $R_{ab}=R^d_{\hphantom{d}adb}$ while the scalar curvature is simply $R\doteq g^{ab}R_{ab}$.

For later convenience we impose a few additional technical constraints on the structure of the admissible manifolds, which we recollect in the following definition:
\begin{definition}
For $n\geq2$, we call the pair $(\Mc,g)$ a {\bf Lorentzian manifold} if $\Mc$ is a Hausdorff, second countable, connected, orientable, smooth manifold $\Mc$ of dimension $n$, endowed with a Lorentzian metric $g$.
\end{definition}
Notice that henceforth we shall always assume that an orientation $\ogth$ for $\Mc$ has been chosen. We could allow in principle more than one connected component, but it would lead to no further insight and, thus, we avoid it for the sake of simplicity. The Lorentzian character of $g$ plays a distinguished role since it entails that all spacetimes come endowed with a causal structure, which lies at the heart of several structural properties of a free quantum field theory. More precisely, let us start from Minkowski spacetime, $\Mc\equiv\RR^4$, endowed with the standard Cartesian coordinates in which the metric tensor reads $\eta_{\mu\nu}=\textrm{diag}(1,-1,-1,-1)$. Let $p\in\RR^4$ be arbitrary. With respect to it, we can split the set of all other points of $\RR^4$ in three separate categories: We say that $q\in\RR^4$ is
\begin{itemize}
\item {\em timelike} separated from $p$ if the connecting vector $v_{pq}$ is such that $\eta(v_{pq},v_{pq})>0$.
\item {\em lightlike} separated if $\eta(v_{pq},v_{pq})=0$
\item {\em spacelike} separated if $\eta(v_{pq},v_{pq})<0$. 
\end{itemize}
If we add to this information the possibility of saying that a point $p$ lies in the future ({\em resp.} in the past) of $q$ if $x^0(p)>x^0(q)$ ({\em resp.} $x^0(p)<x^0(q)$), $x^0$ being the Cartesian time coordinate, we can introduce $I^+_{\RR^4}(p)$ and $I^-_{\RR^4}(p)$, the chronological future $(+)$ and past $(-)$ of $p$, as the collection of all points which are timelike related to $p$ and they lie in its future $(+)$ or in its past $(-)$. Analogously we define $J^\pm_{\RR^4}(p)$ as the causal future $(+)$ and past $(-)$ of $p$ adding also the points which are lightlike related to $p$. Notice that, per convention, $p$ itself is included in both $J_{\RR^4}^+(p)$ and $J_{\RR^4}^-(p)$. In a language more commonly used in theoretical physics, $J^+_{\RR^4}(p)$ and $J^-_{\RR^4}(p)$ are the future and the past light cones stemming from $p$. By extension, if $\Omega$ is an open subregion of $\RR^4$ we call $J^\pm_{\RR^4}(\Omega)\doteq\bigcup_{p\in\Omega}J^\pm_{\RR^4}(p)$. Similarly we define $I^\pm_{\RR^4}(\Omega)$.

On a generic background, the above structures cannot be transported slavishly first of all since, contrary to $\RR^4$, a manifold $M$ does not have look like a Euclidean space globally. In order to circumvent this obstruction, one can start from a generic point $p\in\Mc$ and consider the tangent space $T_p\Mc$. Using the metric $g$, one can label a tangent vector $v\in T_p\Mc$ according to the value of $g(v,v)$. Specifically $v$ is timelike if $g(v,v)>0$, lightlike if $g(v,v)=0$ and spacelike if $g(v,v)<0$. Hence, we can associate to the vector space $T_p\Mc$ a two-folded light cone stemming from $0\in T_p\Mc$ and we have the freedom to set one of the folds as the collection of future-directed vectors. If such choice can be made consistently in a smooth way for all points of the manifold, we say that $\Mc$ is {\em time orientable}. In a geometric language this is tantamount to requiring the existence of a global vector field on $M$ which is timelike at each point. Henceforth we assume that this is indeed the case and that a time orientation $\tgth$ has been fixed. Notice that, as a consequence, every background we consider is completely specified by a quadruple:

\begin{definition}
A {\bf spacetime} $\Mb$ is  a quadruple $(\Mc,g,\ogth,\tgth)$, 
where $(\Mc,g)$ is a time-orientable $n$-dimensional Lorentzian manifold ($n\geq2$), 
$\ogth$ is a choice of orientation on $\Mc$ and $\tgth$ is a choice of time-orientation. 
\end{definition}

The next step in the definition of a causal structure for a Lorentzian manifold consists of considering a piecewise smooth curve $\gamma:I\to\Mc$, $I=[0,1]$. We say that $\gamma$ is timelike ({\em resp.} lightlike, spacelike) if such is the vector tangent to the curve at each point. We say that $\gamma$ is causal if the tangent vector is nowhere spacelike and that it is {\em future (past) directed} if each tangent vector to the curve is future (or past) directed. Taking into account these structures, we can define on an arbitrary spacetime $\Mb$ the chronological future and past of a point $p$ as $I^\pm_\Mb(p)$, the collection of all points $q\in\Mc$ such that there exists a future- (past-)directed timelike curve $\gamma:I\to\Mc$ for which $\gamma(0)=p$ and $\gamma(1)=q$. In complete analogy we can define the causal future and past $J^\pm_\Mb(p)$ as well as, for any open subset $\Omega\subset\Mc$, $J^\pm_\Mb(\Omega)=\bigcup_{p\in\Omega}J^\pm_\Mb(p)$. Similarly we define $I^\pm_\Mb(\Omega)$. We will also denote the union of the causal future $J^+_\Mb(\Omega)$ and the causal past $J^-_\Mb(\Omega)$ of $\Omega$ with $J_\Mb(\Omega)$.

The identification of a causal structure is not only an interesting fingerprint of a spacetime, but it has also far-reaching physical consequences, as it suggests us that not all time-oriented spacetimes should be thought as admissible. As a matter of fact one can incur in pathological situations such as closed timelike and causal curves, which are often pictorially associated to evocative phenomena such as time travel. There are plenty of examples available in the literature ranging from the so-called G\"odel Universe -- see for example \cite{Godel} -- to the Anti-de Sitter spacetime (AdS), which plays nowadays a prominent role in many applications to high energy physics and string theory. Let us briefly sketch the structure of the latter in arbitrary $n$-dimensions, $n>2$ --  see \cite{HE}. AdS$_n$ is a maximally symmetric solution to the Einstein's equations with a negative cosmological constant $\Lambda$. In other words it is a manifold of constant curvature $R=\frac{2n}{n-2}\Lambda$ with the topology $\bS^1\times\RR^{n-1}$. It can be realized in the $(n+1)$-dimensional spacetime $\RR^{n+1}$, endowed with the Cartesian coordinates $x^\mu$, $\mu=0,\dots,n$, and with the metric $\widetilde g=\textrm{diag}(1,1,-1,\dots,-1)$ as the hyperboloid 
\begin{equation*}
\tilde{g}(x,x)=R^2
\end{equation*}
If we consider the locus $x^i=0$, $i=2,\dots,n$ we obtain the circle $(x^0)^2+(x^1)^2=R^2$ together with the induced line element $(dx^0)^2+(dx^1)^2$. In other words we have found a closed curve in AdS$_n$ whose tangent vector is everywhere timelike. 

Even without making any contact with field theory, it is clear that scenarios similar to the one depicted are problematic as soon as one is concerned with the notion of causality. Therefore it is often customary to restrict the attention to a class of spacetimes which avoids such quandary while being at the same time sufficiently large so to include almost all interesting cases. These are the so-called {\it globally hyperbolic spacetimes}. We characterize them following \cite[Chap. 8]{Wald}. As a starting point, we consider a spacetime $\Mb$ and we introduce two additional notions:
\begin{enumerate}
\item A subset $\Sigma\subset\Mc$ is called {\em achronal} if each timelike curve in $\Mb$ intersects $\Sigma$ at most once; %$I^+_\Mb(\Sigma)\cap\Sigma=\emptyset$. Notice that, whenever $\Sigma$ is a submanifold of $\Mc$, in writing $\Sigma\subset\Mc$, it is implicitly assumed that the submanifold comes endowed with the metric and the orientation induced from $\Mb$.
\item For any subset $\Sigma\subset\Mb$, we call future $(+)$, respectively past $(-)$ {\em domain of dependence} $D^\pm_\Mb(\Sigma)$, the collection of all points $q\in\Mc$ such that every past $(+)$, respectively future $(-)$ inextensible causal curve passing through $q$ intersects $\Sigma$. With $D_\Mb(\Sigma)\doteq D^+_\Mb(\Sigma)\cup D^-_\Mb(\Sigma)$ we indicate simply the domain of dependence.
\end{enumerate}
We state the following:
\begin{definition}\label{globhyp}
We say that $\Mb$ is {\bf globally hyperbolic} if and only if there exists a Cauchy surface $\Sigma$, that is a closed achronal subset of $\Mb$ such that $D_\Mb(\Sigma)=\Mc$.
\end{definition}
Notice that, as a by-product of this definition, one can conclude, not only that no closed causal curve exists in $\Mb$, but also that $\Mc$ is homeomorphic to $\RR\times\Sigma$, while $\Sigma$ is a $C^0$, $(n-1)$-dimensional submanifold of $\Mb$, {\it cf.} \cite[Theorem 3.17]{BEE} for the case $n=4$. It is worth mentioning two apparently unrelated points: {\em 1)} in the past, it has been often assumed that the Cauchy surface could be taken as smooth and {\em 2)} Definition \ref{globhyp} does not provide any concrete mean to verify in explicit examples whether a spacetime $\Mb$ is globally hyperbolic or not. Alternative characterizations of global hyperbolicity, such as that $M$ is strongly causal and $J^+_\Mb(p)\cap J^-_\Mb(q)$ is either empty or compact for all $p,q\in\Mc$, did not help in this respect. A key step forward was made a decade ago by the work of Bernal and Sanchez, see \cite{Bernal, Bernal:2005qf2}. Their main result is here stated following the formulation of \cite[Section 1.3]{BGP}:

\begin{theorem}\label{BS}
Let $\Mb$ be given. The following statements are equivalent:
\begin{enumerate}
\item $\Mb$ is globally hyperbolic.
\item There exists no closed causal curve in $\Mb$ and $J^+_\Mb(p)\cap J^-_\Mb(q)$ is either compact or empty for all $p,q\in\Mc$. 
\item $(\Mc,g)$ is isometric to $\RR\times\Sigma$ endowed with the line element $ds^2=\beta dt^2-h_t$, where $t:\RR\times\Sigma\to\RR$ is the projection on the first factor, $\beta$ is a smooth and strictly positive function on $\RR\times\Sigma$ and $t\mapsto h_t$, $t\in\RR$, yields a one-parameter family of smooth Riemannian metrics. Furthermore, for all $t\in\RR$, $\{t\}\times\Sigma$ is an $(n-1)$-dimensional, spacelike, smooth Cauchy surface in $\Mb$. 
\end{enumerate}
\end{theorem}

The main advantage of this last theorem is to provide an easier criterion to verify explicitly whether a given time-oriented spacetime is globally hyperbolic. In order to convince the reader that this class of manifolds includes most of the physically interesting examples, we list a collection of the globally hyperbolic spacetimes which are often used in the framework of quantum field theory on curved backgrounds:

\begin{itemize}
\item We say that a spacetime $\Mb$ is {\em ultrastatic} if $(\Mc,g)$ is isometric to $\RR\times\Sigma$ with line element $ds^2=dt^2-\pi^*h$, where $\pi^*h$ is the pullback along the projection $\pi:\RR\times\Sigma\to\Sigma$ of a metric $h$ on $\Sigma$." $\Mb$ is globally hyperbolic if and only if it is geodesically complete, that is every maximal geodesic is defined on the whole real line -- see \cite{Fulling}. Minkowski spacetime falls in this category.
\item All Friedmann-Robertson-Walker (FRW) spacetimes are four-dimensional homogeneous and isotropic manifolds which are topologically $\RR\times\Sigma$ with 
\begin{equation*}
ds^2=dt^2-a^2(t)\left(\frac{dr^2}{1-kr^2}+r^2d\bS^2(\theta,\varphi)\right),
\end{equation*}
where $d\bS^2(\theta,\varphi)$ is the standard line element of the unit $2$-sphere, while $a(t)$ is a smooth and strictly positive function depending only on time. Furthermore $k$ is a constant which, up to a normalization, can be set to $0$, $1$, $-1$ and, depending on this choice, $\Sigma$ is a three-dimensional spacelike manifold whose model space is either $\RR^3$, the $3$-sphere $\bS^3$ or the three dimensional hyperboloid $\mathbb{H}^3$. The remaining coordinate $r$ has a domain of definition which runs over the whole positive real line if $k=0,-1$, while $r\in (0,1)$ if $k=1$. On account of \cite[Theorem 3.68]{BEE}, we can conclude that every Friedmann-Robertson-Walker spacetime is globally hyperbolic\footnote{We are grateful to Zhirayr Avetisyan for pointing us out Theorem 3.68 in \cite{BEE}.}. Notice that, in many concrete physical applications, the coordinate $t$ runs only on an open interval of $\RR$, but, as proven in \cite[Theorem 3.69]{BEE}, it does not spoil the property of being globally hyperbolic. Following a similar argument one can draw similar conclusions when working with time oriented, homogeneous spacetimes, which are also referred to as Bianchi spacetimes.
\item A noteworthy collection of solutions of the vacuum Einstein's equations consists of the Kerr family which describes a rotating, uncharged, axially symmetric, four-dimensional black hole \cite{Townsend}. In the so-called Boyer-Linquist chart, the underlying line element reads as 
\begin{equation*}
\begin{aligned}
ds^2 = &  \frac{\Delta-a^2\sin^2\theta}{\Pi}dt^2+\frac{4Mar\sin^2\theta}{\Pi}dtd\varphi-\frac{\Pi}{\Delta}dr^2\\
& -\Pi d\theta^2-\frac{(r^2+a^2)^2-\Delta a^2\sin^2\theta}{\Pi}d\varphi^2,
\end{aligned}
\end{equation*}
where $\Delta=r^2-2Mr+a^2$ and $\Pi=r^2+a^2\cos^2\theta$, while $M$ and $J=Ma$ are two real parameters which are interpreted respectively as the mass and total angular momentum of the black hole. Notice that $t$ runs along the whole real line, $\theta,\varphi$ are the standard coordinates over the unit $2$-sphere, while $r$ plays the role of a radial-like coordinate. A generic Kerr spacetime possesses coordinate horizons at $r_\pm=M\pm\sqrt{M^2-a^2}$ provided that $M^2\geq a^2$ and the region for which $r\in(r_+,\infty)$, often also known as exterior region to the black hole, is actually a globally hyperbolic spacetime. If we set $a=0$, that is the black hole does not rotate, we recover the spherically symmetric Schwarzschild spacetime and, consistently, the static region outside the event horizon located at $r=2M$ is itself globally hyperbolic.
\item Another spacetime, which is often used as a working example in quantum field theory on curved spacetime is {\em de Sitter} (dS$_n$), the maximally symmetric solution to the Einstein's equations with a positive cosmological constant $\Lambda$ and $n>2$. In $\RR^{n+1}$ endowed both with the standard Cartesian coordinates $x^\mu$, $\mu=0,\dots,n$, and with the metric $\tilde{g} = \mathrm{diag}(1, -1, \ldots, -1)$, dS$_n$ can be realized as the hyperboloid $\tilde{g}(x,x) = -R^2$, where $R^2=\frac{(n-1)(n-2)}{2\Lambda}$.  After the change of coordinates $x^0=R\sinh(t/R)$ and $x^i=R\cosh(t/R)e^i$, $i=1,\dots,n$, where $\sum_{i=1}^{n}(e^i)^2=1$, the line element of dS$_n$ reads
\begin{equation*}
ds^2=dt^2-R^2\cosh^2(t/R)\,d\bS^2(e^1,\dots,e^n),
\end{equation*}
where $d\bS^2(e^1,\dots,e^n)$ represents the standard line element of the unit $(n-1)$-sphere and $t$ runs along the whole real line. Per direct inspection we see that dS$_n$ is diffeomorphic to an $n$-dimensional version of a Friedmann-Robertson-Walker spacetime with compact spatial sections and it is globally hyperbolic.
\end{itemize}

Since in the next section we will be interested in functions from a globally hyperbolic spacetime to a suitable target vector space and in their support properties, we conclude the section with a useful definition:

\begin{definition}\label{supports}
Let $\Mb$ be a globally hyperbolic spacetime and $V$ a finite dimensional vector space. We call 
\begin{description}
\item[$(0)$] $C^\infty_0(\Mc;V)$ the space of smooth and {\bf compactly supported} $V$-valued functions on $\Mc$,
\item[$(sc)$] $C^\infty_{sc}(\Mc;V)$ the space of smooth and {\bf spacelike compact} $V$-valued functions on $\Mc$, that is $f\in C^\infty_{sc}(\Mc;V)$ if there exists a compact subset $K\subset\Mc$ such that $\supp f\subset J_\Mb(K)$,
\item[$(fc/pc)$] $C^\infty_{fc}(\Mc;V)$ the space of smooth and {\bf future compact} $V$-valued functions on $\Mc$, that is $f\in C^\infty_{fc}(\Mc;V)$ if $\supp f\cap J^+_\Mb(p)$ is compact for all $p\in\Mc$. Mutatis mutandis, we shall also consider $C^\infty_{pc}(\Mc;V)$, namely the space of smooth and {\bf past compact} $V$-valued functions on $\Mc$,
\item[$(tc)$] $C^\infty_{tc}(\Mc;V)\doteq C^\infty_{fc}(\Mc;V)\cap C^\infty_{pc}(\Mc;V)$ the space of smooth and {\bf timelike compact} $V$-valued functions on $\Mc$.
\end{description}
\end{definition} 

\section{On Green hyperbolic operators}\label{sec:2}

Globally hyperbolic spacetimes play a pivotal role, not only because they do not allow for pathological situations, such as closed causal curves, but also because they are the natural playground for classical and quantum fields on curved backgrounds. More precisely, the dynamics of most (if not all) models, we are interested in, is either ruled by or closely related to wave-like equations. Also motivated by physics, we want to construct the associated space of solutions by solving an initial value problem. To this end we need to be able to select both an hypersurface on which to assign initial data and to identify an evolution direction. In view of Theorem \ref{BS}, globally hyperbolic spacetimes appear to be indeed a natural choice. Goal of this section will be to summarize the main definitions and the key properties of the class of partial differential equations, useful to discuss the models that we shall introduce in the next sections. Since this is an overkilled topic, we do not wish to make any claim of being complete and we recommend to an interested reader to consult more specialized books and papers for more details. We suggest for example \cite{Hormander1, Hormander2, Hormander3, Hormander4}, the more recent \cite{Waldmann} and also \cite{BGP}, on which most of this section is based; moreover, notice that, several examples of Green hyperbolic operators can be found in \cite{BG1,BG2}, while a remarkable extension of their domain of definition, which we shall implicitly assume, is available in \cite{Bar}. 

As a starting point we introduce the building block of any classical and quantum field theory:

\begin{definition}\label{VB}
A {\bf vector bundle} of rank $k<\infty$ over an $n$-dimensional smooth manifold $\Mc$ (base space) is $F\equiv F(\Mc,\pi,V)$ where $F$, the total space, is a smooth manifold of dimension $n+k$, $V$ (the typical fiber) is a $k$-dimensional vector space and $\pi:F\to\Mc$ is a smooth surjective map. Furthermore we require that:
\begin{enumerate}
\item There exists a vector space isomorphism between $V$ and each fiber $F_p\doteq \pi^{-1}(p)$, $p\in\Mc$,
\item For each $p\in\Mc$ there exists an open neighbourhood $U\ni p$ and a diffeomorphism $\Psi:\pi^{-1}(U)\to U\times V$ such that $\pi_1\circ\Psi=\pi$ on $\pi^{-1}(U)$, where $\pi_1: U\times V\to U$ is the projection on the first factor of the Cartesian product;
\item The restriction of $\Psi$ to each fiber is an isomorphism of vector spaces.
\end{enumerate}
The pair $(U,\Psi)$ fulfilling these conditions is called a {\bf local trivialization} of $E$.
\end{definition}

Notice that throughout the text we shall use the word {\em vector bundle atlas}, meaning a collection of local trivializations of $F$ covering $\Mc$. %The above definition is not the most general since it would suffice that $\Mc$ were a differentiable manifold, but we prefer to consider only the cases of interest to us. 
We will not discuss the theory of vector bundles and, for more details, refer to \cite{Husemoller}. The only exceptions are the following two definitions:

\begin{definition}\label{restriction}
Let $F=F(\Mc,\pi,V)$ be a vector bundle and let $N$ be a submanifold of $\Mc$. We call {\bf restriction} of $F$ to $N$ the vector bundle $F\vert_N\equiv \widetilde F=\widetilde F(N,\pi^\prime,V)$, where $\widetilde F=\pi^{-1}(N)$ and $\pi^\prime:\widetilde{F} \to N$ is defined by $\pi^\prime(f) = \pi(f)$ for all $f \in \widetilde{F}$.
\end{definition}

\begin{definition}\label{dualbundle}
Let $F=F(\Mc,\pi,V)$ be a vector bundle. We call {\bf dual bundle} $F^*$ the vector bundle over $\Mc$ whose fiber over $p\in\Mc$ is $(F^*)_p=(F_p)^*$, the dual vector space to $F_p$.
\end{definition}

We say that a vector bundle $F$ is {\em (globally) trivial} if there exists a fiber preserving diffeomorphism from $F$ to the Cartesian product $\Mc\times V$ restricting to a vector space isomorphism on each fiber. In practice, this corresponds to a trivialization of $F$ which is defined everywhere, to be compared with the notion of a local trivialization as per Definition \ref{VB}. Most of the examples we shall consider in this paper come from globally trivial vector bundles. Bear in mind, however, that one of the canonical examples of a vector bundle, namely the tangent bundle $T\Mc$ to a manifold $\Mc$, is not trivial in general, {\it e.g.}, when $\Mc=\bS^2$. It is also noteworthy that, given any two vector bundles $F=F(\Mc,\pi,V)$ and $F^\prime=F^\prime(\Mc,\pi^\prime,V^\prime)$, we can construct naturally a third vector bundle, the {\em bundle of homomorphisms} $\Hom(F,F^\prime)$ over the base space $\Mc$. Its fiber over a base point $p\in\Mc$ is $\Hom(F_p,F^\prime_p)$, which is a vector space isomorphic to the vector space $\Hom(V,V^\prime)$ of homomorphism from $V$ to $V^\prime$. If $F^\prime=F$, then we shall write $\Endo(F)$ for $\Hom(F,F^\prime)$ and call it {\em bundle of endomorphisms}, whose typical fiber is $\Endo(V)$.

Another structure which plays a distinguished role is the following:
\begin{definition}
Given a vector bundle $F$, $\Gamma(F)=\{\sigma\in C^\infty(\Mc;F)\;|\;\pi\circ\sigma=\id_\Mc\}$, where $\id_\Mc:\Mc\to\Mc$ is the identity map, is the {\bf space of smooth sections}. By generalizing Definition \ref{supports}, the subscripts $0$, $sc$, $fc/pc$ and $tc$ shall refer to those sections whose support is respectively compact, spacelike compact, future or past compact and timelike compact.
\end{definition}

Notice that $\Gamma(F)$ is an infinite-dimensional vector space and, whenever $F$ is trivial, it is isomorphic to $C^\infty(\Mc;V)$. The next structure, which we introduce, will play an important role in the construction of an algebra of observables for a free quantum field on a curved background:
\begin{definition}\label{pairing}
Let $F$ be a vector bundle on a manifold $\Mc$. Denote with $F\times_\Mc F$ the bundle obtained taking the Cartesian product fiberwise. A non-degenerate inner product on $F$ is a smooth map $\cdot:F\times_\Mc F\to\RR$ such that 
\begin{enumerate}
\item the restriction of $\cdot$ to $F_p\times F_p$ is a bilinear form for each $p\in\Mc$, 
\item $v\in F_p$ vanishes if $v\cdot w=0$ for all $w\in F_p$.
\end{enumerate}
Furthermore, if $\cdot$ is symmetric on each fiber, we call it {\bf Bosonic}, if it is antisymmetric, we call it {\bf Fermionic}.
\end{definition}

Since we consider only spacetimes $\Mb=(\Mc,g,\ogth,\tgth)$, the orientation of $\Mc$ is fixed by $\ogth$ and therefore we can introduce the metric-induced volume form $\dvol_\Mb$ on $\Mc$. Then, any inner product as in Definition \ref{pairing} induces a non-degenerate pairing between smooth sections and compactly supported smooth sections of $F$:
\begin{align}\label{pairing-sections}
\ipf{\cdot}{\cdot}:\Gamma_0(F)\times\Gamma(F)\to\RR, && (\sigma,\tau)\mapsto\int_\Mc(\sigma\cdot\tau)\,\dvol_\Mb.
\end{align} 
Since we will make use of it later in this paper, notice that, for \eqref{pairing-sections} to be meaningful, we could consider both $\tau,\sigma\in\Gamma(F)$ provided that $\supp(\tau)\cap\supp(\sigma)$ is compact. 

We have all ingredients to start addressing the main point of this section, namely partial differential equations. The building block is the following:
\begin{definition}\label{PDE}
Let $F=F(\Mc,\pi,V)$ and $F^\prime=F^\prime(\Mc,\pi^\prime,V^\prime)$ be two vector bundles of rank $k$ and $k^\prime$ respectively over the same manifold $\Mc$. A {\bf linear partial differential operator} of order at most $s\in\NN_0$ is a linear map $L:\Gamma(F)\to\Gamma(F^\prime)$ such that, for all $p\in\Mc$, there exist both a coordinate neighborhood $(U,\Phi)$ centered at $p$, local trivializations $(U,\Psi)$ and $(U,\Psi^\prime)$ respectively of $F$ and of $F^\prime$, as well as a collection of smooth maps $A_\alpha:U\to \Hom(V;V^\prime)$ labeled by multi-indices for which, given any $\sigma\in\Gamma(F)$, on $U$ one has 
\begin{equation*}
L\sigma=\sum_{|\alpha|\leq s}A_\alpha\,\partial^\alpha\sigma.
\end{equation*}
Notice that here we are implicitly using both the coordinate chart $\Phi$ and the trivializations $\Psi$ and $\Psi^\prime$; moreover, the sum runs over all multi-indices $\alpha=(\alpha_0,\dots,\alpha_{n-1})\in\NN_0^n$ such that$|\alpha|\doteq\sum_{\mu=0}^{n-1}\alpha_\mu\leq s$ and $\partial^\alpha=\prod_{\mu=0}^{n-1}\partial_\mu^{\alpha_\mu}$, where $\partial_0,\dots,\partial_{n-1}$ are the partial derivatives with respect to the coordinates $x^0,\dots,x^{n-1}$ coming from the chart $(U,\Phi)$. Furthermore, $L$ is of order $s\in\NN_0$ when it is of order at most $s$, but not of order at most $s-1$. 
\end{definition}

Notice that linear partial differential operators cannot enlarge the support of a section, a property which will be often used in the rest of this paper. Another related and useful concept intertwines linear partial differential operators with the pairing \eqref{pairing-sections}:

\begin{definition}\label{adjoint}
Consider a spacetime $\Mb=(\Mc,g,\ogth,\tgth)$. 
Let $F=E(\Mc,\pi,V)$ and let $F^\prime=F^\prime(\Mc,\pi^\prime,V^\prime)$ be two vector bundles over the manifold $\Mc$, both endowed with a non-degenerate inner product. Denote the pairings defined in \eqref{pairing-sections} 
for $F$ and for $F^\prime$ respectively with $\ipf{\cdot}{\cdot}_F$ and $\ipf{\cdot}{\cdot}_{F^\prime}$. Let $L:\Gamma(F)\to\Gamma(F^\prime)$ be a linear partial differential operator. We call {\bf formal adjoint} of $L$ the linear partial differential operator $L^*:\Gamma(F^\prime)\to\Gamma(F)$ such that, for all $\sigma\in\Gamma(F)$ and $\tau\in\Gamma(F^\prime)$ with supports having compact overlap, the following identity holds: 
\begin{equation*}
\ipf{L^*\tau}{\sigma}_{F^\prime}=\ipf{\tau}{L\sigma}_F.
\end{equation*}
If $F=F^\prime$, we say that $L$ is {\bf formally self-adjoint} whenever $L^*$ coincides with $L$.
\end{definition}
Existence of $L^*$ is a consequence of Stokes theorem and uniqueness is instead due to the non-degeneracy of the pairing $\ipf{\cdot}{\cdot}_{F^\prime}$. 

Definition \ref{PDE} accounts for a large class of operators, most of which are not typically used in the framework of field theory, especially because they cannot be associated to an initial value problem. In order to select a relevant class for our purposes, we introduce a useful concept -- see also \cite{BGP, Hormander1}:

\begin{definition}\label{normhyp}
Let $F=F(\Mc,\pi,V)$ and $F^\prime=F^\prime(\Mc,\pi^\prime,V^\prime)$ be two vector bundles over the same manifold $\Mc$ and let $L:\Gamma(F)\to\Gamma(F^\prime)$ be any linear partial differential operator of order $s$ as per Definition \ref{PDE}. We call {\bf principal symbol} of $L$ the map $\sigma_L:T^*M\to \Hom(F,F^\prime)$ locally defined as follows: For a given $p\in\Mc$, mimicking Definition \ref{PDE}, consider a coordinate chart around $p\in\Mc$ and local trivializations of $F$ and of $F^\prime$ and, for all $\zeta\in T^\ast_p\Mc$, set 
\begin{equation*}
\sigma_L(\zeta)=\sum_{|\alpha|=s}A_\alpha(p)\zeta^\alpha,
\end{equation*}
where $\zeta^\alpha=\prod_{\mu=0}^{n-1}\zeta_\mu^{\alpha_\mu}$ and $\zeta_\mu$ are the components of $\zeta$ with respect to the chosen chart. Furthermore, given a Lorentzian manifold $(\Mc,g)$, we call a second order linear partial differential operator $P:\Gamma(F)\to\Gamma(F^\prime)$ {\bf normally hyperbolic} if $\sigma_P(\zeta)=g(\zeta,\zeta)\,\id_{F_p}$ for all $p\in\Mc$ and all $\zeta\in T^\ast_p\Mb$.
\end{definition}

In order to better grasp the structure of a normally hyperbolic operator $P$, we can write it in a local coordinate frame following both Definition \ref{PDE} and Definition \ref{normhyp}. Let thus $p$ be in $\Mc$ and $(U,\Phi)$ be a chart centered at $p$ where the vector bundle $F$ is trivial. There exist both $A$ and $A^\mu$, $\mu=0,\dots,n-1$, smooth maps from $U$ to $\Endo(V)$ such that, for any $\sigma\in\Gamma(F)$, on $U$ one has 
\begin{equation*}
P\sigma= g^{\mu\nu}\id_V\,\partial_\mu\partial_\nu\sigma+A^\mu\,\partial_\mu\sigma+A\,\sigma,
\end{equation*}
where both the chart and the vector bundle trivializations are understood. One immediately notices that locally this expression agrees up to terms of lower order in the derivatives with the one for the d'Alembert operator acting on sections of $F$ constructed out of a covariant derivative $\nabla$ on $F$, that is the operator $\Box_\nabla=g^{\mu\nu}\nabla_\mu\nabla_\nu:\sect(F)\to\sect(F)$. Therefore, one realizes that normally hyperbolic operators provide the natural generalization of the usual d'Alembert operator. Besides this remark, Definition \ref{normhyp} becomes even more important if we assume, moreover, that the underlying background is globally hyperbolic, since we can associate to each normally hyperbolic operator $P$ an initial value problem (also known as Cauchy problem). As a matter of fact, in view both of Definition \ref{globhyp} and of Theorem \ref{BS}, initial data can be assigned on each Cauchy surface and the following proposition shows the well-posedness of the construction:

\begin{proposition}\label{Cauchyproblem}
Let $\Mb=(\Mc,g,\ogth,\tgth)$ be a globally hyperbolic spacetime and $\Sigma\subset\Mc$ any of its spacelike Cauchy surfaces together with its future-pointing unit normal vector field $\nb$. Consider any vector bundle $F$ over $\Mc$, a normally hyperbolic operator $P:\Gamma(F)\to\Gamma(F)$, and a $P$-compatible\footnote{A covariant derivative $\nabla$ on $F$ is $P$-compatible if there exists a section $A\in\sect(\Endo(F))$ such that $\Box_\nabla+A=P$.} covariant derivative $\nabla$ on $F$. Let $F\vert_\Sigma$ be the restriction of $F$ to $\Sigma$ as per Definition \ref{restriction}. Then, for any $J\in\Gamma(F)$ and for any $u_0,u_1\in\Gamma(F\vert_\Sigma)$, the following initial value problem admits a unique solution $u\in\Gamma(F)$: 
\begin{equation}\label{inhom}
\begin{aligned}
Pu & =  J && \mbox{on $\Mc$},\\
u & =  u_0 && \mbox{on $\Sigma$},\\
\quad \nabla_\nb u & =  u_1 && \mbox{on $\Sigma$}.
\end{aligned}
\end{equation}
Furthermore, if we set $\Omega=\supp u_0\cup\supp u_1\cup\supp J$, then $\supp u \subset J_\Mb(\Omega)$.
\end{proposition}

The proof of this proposition has been given in different forms in several books, {\it e.g.} \cite{BGP, Friedlander} and in \cite[Corollary 5]{Bar:2009zzb}. Notice that equation \eqref{inhom} is not linear since we allow for a non-vanishing source term. For pedagogical reasons, we shall henceforth consider only the case $J=0$ although the reader should keep in mind that such constraint is not really needed and a treatment especially of quantization in this scenario has been given in \cite{BDS14affine} and further refined in \cite{FS14}. 

The characterization of all smooth solutions of the equation $Pu=0$ represents the first step in outlining a quantization scheme for a free field theory. To this end one does not consider directly the Cauchy problem \eqref{inhom}, but rather exploits a notable property of normally hyperbolic operators, namely the fact that on any globally hyperbolic spacetime they come together with Green operators. Here we introduce them paying particular attention to the domain where they are defined. The reader should keep in mind that our presentation slightly differs in comparison for example to \cite{BGP} and we make use of results which are presented in \cite{Bar, Friedlander, Sanders:2012}:

\begin{definition}\label{E+E-}
Let $\Mb=(\Mc,g,\ogth,\tgth)$ be a globally hyperbolic spacetime and consider a vector bundle $F$ over $\Mc$. Furthermore, let $L:\Gamma(F)\to\Gamma(F)$ be a linear partial differential operator. We call {\bf retarded} $(+)$ and {\bf advanced} $(-)$ {\bf Green operators} two linear maps
\begin{align}
E^+:\Gamma_{pc}(F)\to\Gamma(F), && E^-:\Gamma_{fc}(F)\to\Gamma(F),
\end{align}
fulfilling the properties listed below: 
\begin{enumerate}
\item For any $f\in\Gamma_{pc}(F)$, it holds $LE^+f=f=E^+Lf$ and $\supp(E^+f)\subset J^+_\Mb(\supp f)$,
\item For any $f\in\Gamma_{fc}(F)$, it holds $LE^-f=f=E^-Lf$ and $\supp(E^-f)\subset J^-_\Mb(\supp f)$.
\end{enumerate}
The operator $E\doteq E^--E^+:\Gamma_{tc}(F)\to\Gamma(F)$ will be referred to as {\bf advanced-minus-retarded operator}. A linear partial differential operator admitting both $E^+$ and $E^-$ will be called {\bf Green hyperbolic}.
\end{definition}

Notice that in the literature the symbols $E^\pm$ are often written as $G^\pm$, while the operator $E$ is also called {\em causal propagator}. We avoid this nomenclature since, from time to time, it is also used for completely different objects and we wish to avoid a potential source of confusion for the reader. In view of the application of this material to some specific field theoretical models, see Section \ref{sect:cqft}, we introduce now the canonical integral pairing $\langle\cdot,\cdot\ra$ between the sections of a vector bundle $F$ and those of its dual $F^\ast$. This is defined by integrating over the base manifold the fiberwise pairing between $F^\ast$ and $F$:
\begin{equation}\label{intrinsicdual}
\langle f^\prime,f\ra\doteq\int_{\Mc}f^\prime(f)\,\dvol_\Mb,
\end{equation}
where $f\in\sect(F)$ and $f^\prime\in\sect(F^\ast)$ have supports with compact overlap. 
Notice that the formula above provides non-degenerate bilinear pairings between $\sectc(F^\ast)$ and $\sect(F)$ 
and between $\sect(F^\ast)$ and $\sectc(F)$. 
In turn, this pairing allows for the following definition:
\begin{definition}
Let $\Mb=(\Mc,g,\ogth,\tgth)$ be a globally hyperbolic spacetime and consider a vector bundle $F$ over $\Mc$ and its dual $F^\ast$. Furthermore, let $L:\Gamma(F)\to\Gamma(F)$ be a partial differential operator. We call {\bf formal dual} the linear partial differential operator $L^\star:\Gamma(F^*)\to\Gamma(F^*)$ defined through \eqref{intrinsicdual} via 
\begin{align}
\langle L^\star f^\prime,f\rangle=\langle f^\prime,Lf\rangle,
\end{align}
where $f\in\sect(F)$ and $f^\prime\in\sect(F^\ast)$ have supports with compact overlap.
\end{definition}

\noindent Notice that, if $F$ is endowed with a non-degenerate inner product as per Definition \ref{pairing}, then we can identify $F$ with $F^*$. Via this identification \eqref{intrinsicdual} becomes \eqref{pairing-sections}. 

\begin{proposition}\label{uniqueG}
Let $\Mb=(\Mc,g,\ogth,\tgth)$ be a globally hyperbolic spacetime and consider a vector bundle $F$ over $\Mc$ and its dual $F^\ast$. Furthermore, let both $L:\Gamma(F)\to\Gamma(F)$ and its formal dual $L^\star:\Gamma(F^*)\to\Gamma(F^*)$ be Green hyperbolic operators, whose retarded and advanced Green operators are $E^\pm$ and $E^{\star\pm}$ respectively. Then, for all $f\in\Gamma_0(F)$ and $f^\prime\in\Gamma_0(F^*)$, it holds that 
\begin{equation*}
\lan E^{\star\mp}f^\prime,f\ra=\lan f^\prime,E^\pm f\ra,
\end{equation*}
where $\lan\cdot,\cdot\ra$ is the pairing defined in \eqref{intrinsicdual}.
\end{proposition}

\begin{proof}
The statement is a consequence of the following chain of equalities, which holds true for arbitrary $f^\prime\in\Gamma_0(F^\star)$ and $f\in\Gamma_0(F)$:
\begin{equation*}
\lan E^{\star\mp}f^\prime,f\ra=\lan E^{\star\mp}f^\prime,LE^\pm f\ra=\lan L^\star E^{\star\mp}f^\prime,E^\pm f\ra=\lan f^\prime, E^\pm f\ra.
\end{equation*}
\end{proof}

From the definition, $L$ admits left-inverses $E^+$ and $E^-$ on sections with past (respectively future) compact support. In other words $L$ is injective thereon. As a consequence, $E^+$ and $E^-$ are uniquely specified by their support properties and by the condition of being also right-inverses of L on $\Gamma_{pc}(F)$ and respectively on $\Gamma_{fc}(F)$.

\begin{lemma}\label{uniqueGlemma}
Let $\Mb=(\Mc,g,\ogth,\tgth)$ be a globally hyperbolic spacetime. Consider a vector bundle $F$ over $\Mc$ endowed with a non-degenerate inner product as per Definition \ref{pairing}. Furthermore, let both $L:\Gamma(F)\to\Gamma(F)$ and its formal adjoint $L^*:\Gamma(F)\to\Gamma(F)$ be Green hyperbolic operators. Then, calling $E^\pm$ and $E^{\ast\pm}$ their respective retarded and advanced Green operators, the identity
\begin{equation*}
(E^{\ast\mp}f^\prime,f)=(f^\prime, E^\pm f)
\end{equation*}
holds for all $f^\prime\in\Gamma_0(F)$ and $f\in\Gamma_0(F)$.
\end{lemma}

\begin{proof}
Since $F$ is endowed with a non-degenerate inner product, \eqref{intrinsicdual} reduces to \eqref{pairing-sections} upon identification of $F^*$ with $F$. Furthermore the formal dual of $L$ coincides with its formal adjoint under this identification, {\it i.e.} $L^\star=L^*$. Hence we are falling in the hypotheses of Proposition \ref{uniqueG}, from which the sought result follows.
\end{proof}

Notice that all normally hyperbolic operators are Green hyperbolic. This result follows from \cite{Bar, BGP}. In the latter reference the retarded and advanced Green operators for a normally hyperbolic operator are shown to exist, although with a smaller domain compared to the one we consider here, while in the first one the domains are uniquely extended, thus fulfilling the requirement of our Definition \ref{E+E-}. Let us stress that there are physically interesting partial differential operators which are Green hyperbolic, but not normally hyperbolic. The most notable example is the Dirac operator which will be discussed in Section \ref{subDirac}. The reader should also keep in mind that some authors are calling Green hyperbolic an operator which fulfills the hypotheses of Proposition \ref{uniqueG} or of Lemma \ref{uniqueGlemma} -- see for example \cite{Bar}.

The usefulness of both retarded and advanced Green operators becomes manifest as soon as one notices that, to every timelike compact section $f$ of a vector bundle $F$ we can associate a solution of the linear equation $Lu=0$ as $u=Ef=E^-f-E^+f$. Yet, before concluding that we have given a characterization of all solutions, we need a few additional data:

\begin{lemma}\label{Kernel}
Let $\Mb=(\Mc,g,\ogth,\tgth)$ be a globally hyperbolic spacetime. 
Consider a vector bundle $F$ over $\Mc$ and a Green hyperbolic operator $L:\Gamma(F)\to\Gamma(F)$. Let $E^\pm$ be the retarded and advanced Green operators for $L$ and denote with $E$ the corresponding advanced-minus-retarded operator. Then $f\in\secttc(F)$ is such that $Ef=0$ if and only if $f=Lh$ for some $h\in\Gamma_{tc}(F)$. Furthermore, $f\in\Gamma_{tc}(F)$ is such that $Lf=0$ if and only if $f=0$ and, moreover, for any $f\in\sect(F)$ there exists $h\in\sect(F)$ such that $Ph=f$.
\end{lemma}

\begin{proof}
On account of Definition \ref{E+E-}, it holds that $ELh=0$ for all $h\in\secttc(F)$, thus we need only to show that, given $f\in\Gamma_{tc}(F)$ such that $Ef=0$, then there exists $h\in\Gamma_{tc}(F)$ such that $f=Lh$. Taking any such $f$, $Ef=0$ implies that $E^-f=E^+f$. The support properties of the retarded and advanced Green operators entail that $\supp(E^-f)\subset J^+(\supp f)\cap J^-(\supp f)$. In other words $h=E^-f\in\Gamma_{tc}(F)$. If we apply the operator $L$, it holds $Lh=LE^-f=f$. 

Suppose now that there exists $f\in\Gamma_{tc}(F)$ such that $Lf=0$. By applying either the retarded or the advanced Green operators we obtain, $f=E^\pm Lf=0$.

To conclude the proof, consider $f\in\sect(F)$. Taking a partition of unity $\{\chi_+,\chi_-\}$ on $\Mb$ such that 
$\chi_\pm=1$ on a past/future compact region,\footnote{A partition of unity such as the one described exists on account of Theorem \ref{BS}. In fact, after splitting the globally hyperbolic spacetime $\Mb$ in the Cartesian product of $\RR$ and a spacelike Cauchy surface $\Sigma$ and for any choice of $t_\pm\in\RR$ with $t_-<t_+$, one can introduce a partition of unity $\{\chi_+,\chi_-\}$ on $\RR$ such that $\chi_\pm(t)=1$ for $\pm t\leq\pm t_\pm$. 
Pulling this partition of unity back to $\Mc$ along the the projection on the time factor $t:\Mc\to\RR$, one obtains a partition of unity on $\Mc$ of the sought type.} one can introduce $h=E^+(\chi_+f)+E^-(\chi_-f)\in\sect(F)$. Since $\chi_++\chi_-=1$ everywhere, $Ph=f$ as claimed. 
\end{proof}

In view of this last result, we can finally characterize the space of solutions of $Lu=0$ via the advanced-minus-retarded operator -- see also \cite{Wald2}:

\begin{theorem}\label{main}
Let $\Mb=(\Mc,g,\ogth,\tgth)$ be a globally hyperbolic spacetime. 
Consider a vector bundle $F$ over $\Mc$ and let $L:\Gamma(F)\to\Gamma(F)$ be a Green hyperbolic operator. Let $E^\pm$ be the retarded and advanced Green operators for $L$ and denote with $E$ the corresponding advanced-minus-retarded operator. The map presented below is a vector space isomorphism between $\Sol$, the vector space of smooth solutions of the linear partial differential equation $Lu=0$, $u\in\Gamma(F)$, and the quotient of $\Gamma_{tc}(F)$ by the image of $L$ acting on $\Gamma_{tc}(F)$:
\begin{align}
\frac{\secttc(F)}{L(\secttc(F))}\to\Sol, && [f]\mapsto Ef,
\end{align}
where $f\in\secttc(F)$ is any representative of the equivalence class $[f]$ in the quotient space. 
\end{theorem}

\begin{proof}
Let us notice that the advanced-minus-retarded operator $E:\Gamma_{tc}(F)\to\Gamma(F)$ induces the sought map from $\Gamma_{tc}(F)/L(\Gamma_{tc}(F))$ to $\Sol$. On account of Lemma \ref{Kernel} the image does not depend on the representative of $[f]$ and $u=Ef$ is a solution of $Pu=0$. This map is injective since, given $f,f^\prime\in\Gamma_{tc}(E)$ such that $Ef=Ef^\prime$, per linearity of $E$ and applying Lemma \ref{Kernel}, one finds $h\in\Gamma_{tc}(E)$ such that $Lh=f-f^\prime$. In other words $f$ and $f^\prime$ are two representatives of the same equivalence class in $\Gamma_{tc}(F)/L(\Gamma_{tc}(F))$, which entails injectivity. Only surjectivity is still to be proven. Given $u\in\Sol$ and taking into account a partition of unity $\{\chi_+,\chi_-\}$ on $\Mc$ such that $\chi_\pm=1$ in a past/future compact region, one finds $L(\chi_+u+\chi_-u)=Lu=0$, 
therefore $h=L(\chi_-u)=-L(\chi_+u)$ is timelike compact. Exploiting the properties of retarded and advanced Green operators, one concludes the proof: 
\begin{equation*}
Eh=E^-L(\chi_-u)+E^+L(\chi_+u)=\chi_-u+\chi_+u=u.
\end{equation*}
\end{proof}

It might be useful to summarize the content of Lemma \ref{Kernel} and of Theorem \ref{main} with the following exact sequence:
\begin{equation*}
0\longrightarrow \Gamma_{tc}(F)\overset{L}{\longrightarrow} \Gamma_{tc}(F)\overset{E}{\longrightarrow} \Gamma(F)\overset{L}{\longrightarrow} \Gamma(F)\longrightarrow 0.
\end{equation*}
We remind the reader that this is simply a symbolic way of stating that the kernel of each of the arrows depicted above coincides with the image of the preceding one. 

%a Cauchy problem whose initial data are the restriction to a Cauchy surface $\Sigma$ of $u$ and of $\nabla_{\mathfrak{n}} u$, $\mathfrak{n}$ being the normal to $\Sigma$

%Let thus $u\in\mathcal{S}_L(\Mb)$, let $t$ be the same time coordinate as it appears in Theorem \ref{BS} and let $\chi\equiv\chi(t)$ be a smooth function on $\Mc$ which vanishes for all $t<t_0$, whereas $\chi=-1$ for all $t>t_1$. Since $f\doteq L\chi u\in\Gamma_{tc}(F)$, we can compute $E(f)=E^-(f)-E^+(f)$. Per construction $\chi u\in\Gamma_{pc}(F)$ and thus $E^+(f)=(E^+\circ L)\chi u=\chi u$. On the other hand, since $(1+\chi)u\in\Gamma_{fc}(F)$ and $L(u)=0$ per hypothesis, we can write $f=L (1+\chi)u$ and, consequently, $E^-(f)=(1+\chi)u$ from which it descends $E(f)=u$, namely the sought result.

Although the last theorem provides a complete characterization of the solutions of the partial differential equation associated to a Green hyperbolic operator, we need to introduce and to study a vector subspace of $\Sol$ which will play a distinguished role in the analysis of explicit models. 

\begin{proposition}\label{prpSCSol}
Let $\Mb=(\Mc,g,\ogth,\tgth)$ be a globally hyperbolic spacetime. 
Consider a vector bundle $F$ over $\Mc$ and let $L:\Gamma(F)\to\Gamma(F)$ be a Green hyperbolic operator. Let $E^\pm$ be the retarded and advanced Green operators for $L$ and denote with $E$ the corresponding advanced-minus-retarded operator. Then the following statements hold true: 
\begin{enumerate}
\item If $f\in\sectc(F)$ is such that $Lf=0$, then $f=0$; 
\item If $f\in\Gamma_0(F)$ is such that $Ef=0$, then there exists $h\in\sectc(F)$ such that $Lh=f$;
\item For each $h\in\sectsc(F)$ there exists $f\in\sectsc(F)$ such that $Lf=h$.
\end{enumerate}
Furthermore, let $\Solsc \subset \Sol$ be the vector subspace whose elements are smooth and spacelike compact solutions of $Lu=0$. Then the map presented below is an isomorphism between $\Solsc$ and the quotient of $\Gamma_0(F)$ by the image of $L$ acting on $\Gamma_0(F)$:
\begin{align}
\frac{\sectc(F)}{L(\sectc(F))}\to\Solsc, && [f]\mapsto Ef,
\end{align}
where $f\in\sectc(F)$ is any representative of the equivalence class $[f]$ in the quotient space. 
\end{proposition}

\begin{proof}
The proof follows slavishly those of Lemma \ref{Kernel} and of Theorem \ref{main} and therefore we shall not repeat it in details. One has only to keep in mind that $E$ maps sections with compact support to sections with spacelike compact support and that the intersection between a spacelike compact region and a timelike compact one is compact.
\end{proof}

\noindent In terms of an exact sequence, this last proposition translates to %-- see for example \cite{Waldmann}
\begin{equation}\label{eqSCExactSeq}
0\longrightarrow \Gamma_{0}(F)\overset{L}{\longrightarrow} \Gamma_{0}(F)\overset{E}{\longrightarrow} \Gamma_{sc}(F)\overset{L}{\longrightarrow} \Gamma_{sc}(F)\longrightarrow 0.
\end{equation}

Spacelike compact solutions to a linear partial differential equation are also noteworthy since, under certain additional assumptions, they can be naturally endowed with an additional structure which plays a key role in the construction of the algebra of observables for a Bosonic of for a Fermionic free quantum field theory -- see also \cite{BGP, Hack:2012dm, Khavkine:2014kya, Khavkine:2012jf, Wald2}.

\begin{proposition}\label{prpSymplStructure}
Let $\Mb=(\Mc,g,\ogth,\tgth)$ be a globally hyperbolic spacetime. 
Consider a vector bundle $F$ over $\Mc$ endowed with a non-degenerate inner product as per Definition \ref{pairing}. Let $L:\Gamma(F)\to\Gamma(F)$ be a formally self-adjoint Green hyperbolic operator and denote with $E^\pm$ the corresponding retarded and advanced Green operators and with $E$ the associated advanced-minus-retarded operator. Then the map presented below defines a non-degenerate bilinear form on $\sectc(F)$:
\begin{align}
\tau:\frac{\Gamma_0(F)}{L(\Gamma_0(F))}\times \frac{\Gamma_0(F)}{L(\Gamma_0(F))}\to\RR,
&& ([f],[f^\prime])\mapsto(f,Ef^\prime)%=\int_\Mb(f\cdot(Ef^\prime))\,\dvol_\Mb,
\end{align}
where $(\cdot,\cdot)$ is the pairing defined in \eqref{pairing-sections}, while $f\in[f]$ and $f^\prime\in[f^\prime]$ are two arbitrary representatives. Furthermore, $\tau$ is a symplectic form in the Bosonic case, namely when the inner product on $F$ is symmetric, while it is a scalar product in the Fermionic case, namely when the inner product on $F$ is anti-symmetric.
\end{proposition}

\begin{proof}
Notice that the definition of the map $\tau$ is well-posed since it does not depend on the choice of representatives. 
In fact, on the one hand, $ELh=0$ for all $h\in\sectc(F)$ and, on the other hand, $L$ is formally self-adjoint. 
From its definition, it immediately follows that $\tau$ is bilinear. Let us show that it is also non-degenerate. 
Suppose $f\in\sectc(F)$ is such that $\tau([f],[f^\prime])=0$ for all $f^\prime\in\sectc(F)$. 
This means that $-(Ef,f^\prime)=(f,Ef^\prime)=0$ for all $f^\prime\in\sectc(F)$. 
Here we exploited the fact that $L$ is formally self-adjoint, 
therefore Lemma \ref{uniqueGlemma} holds with $L^\ast=L$. 
Since the pairing $(\cdot,\cdot)$ between $\sect(F)$ and $\sectc(F)$ is non-degenerate, one deduces that $Ef=0$. 
Recalling also \eqref{eqSCExactSeq}, it follows that $f$ lies in $L(\sectc(F))$, meaning that $[f]=0$. 
Similarly, one can show non-degeneracy in the other argument too. 
To conclude the proof, suppose that we are in the Bosonic (Fermionic) case, 
namely we have $a\cdot b=(-)b\cdot a$ for all $p\in\Mc$ and all $a,b\in F_p$. 
Therefore, $\tau_\Mb$ is anti-symmetric (symmetric) as the following chain of identities shows:
\begin{align}
-\tau([f],[f^\prime])=-(f,Ef^\prime)=(Ef,f^\prime)=(-)(f^\prime,Ef)=(-)\tau([f^\prime],[f]).
\end{align}
Note that we exploited the formal self-adjointness of $L$ in the first place 
and then also the symmetry (anti-symmetry) of $\cdot$. 
\end{proof}

Notice that, in the literature it is also customary to denote $\tau([f],[f^\prime])$ with $E(f,f^\prime)$. 

\section{Classical and quantum field theory}\label{sect:cqft}
In this section we shall construct the classical field theories and their quantum counterparts for three models, 
namely the real scalar field, the Proca field and the Dirac field. 
We shall consider an arbitrary, but fixed, globally hyperbolic spacetime $\Mb=(\Mc,g,\ogth,\tgth)$ 
as the background for the dynamical evolution of the fields under analysis. 
For each model, we shall introduce a suitable class of sufficiently well-behaved functionals 
defined on the space of classical field configurations. The goal is to find functionals 
which can be thought of as {\em classical observables} in the sense that 
one can extract any information about a given field configuration by means of these functionals 
and, moreover, each of them provides some information which cannot be detected by any other functional. 
Note that the approach we adopt allows for extensions in several directions. In fact, it has been followed 
both in the context of affine field theories \cite{BDS14affine} as well as for gauge field theories 
\cite{Benini:2013tra, BDHS14, BDM14}. 
Even when a space of classical observables complying with these requirements has been found, 
a symplectic structure is still needed in order to have the full data describing our classical field theory. 
This structure will be induced in a natural way by the partial differential equation ruling the dynamics. 
The reasons for the need of a symplectic structure are manifold. 
Conceptually, the analogy with classical mechanics motivates chiefly this requirement; 
at a practical level, instead, this is a bit of information which is needed to step-up a quantization scheme for the models, we are interested in. 

\subsection{The real scalar field}\label{Trsf}
As mentioned above, let us fix once and for all a globally hyperbolic spacetime $\Mb=(\Mc,g,\ogth,\tgth)$, which provides the background where to specify the field equation. 

\subsubsection{Classical field theory}\label{subsubObsScalar}
For the real scalar case, the {\em off-shell} field configurations are real-valued smooth functions on $\Mc$. 
This means that, before imposing the field equation, the relevant space of configurations is $\func(\Mc)$. 
As a starting point, we introduce linear functionals on $\func(\Mc)$ as follows: Given $f\in\cc(\Mc)$, 
we denote by $F_f:\func(\Mc,\RR)\to\RR$ the map defined below: 
\begin{equation}\label{eqFunctScalar}
F_f(\phi)=\int_\Mc f\;\phi\,\dvol_\Mb,
\end{equation}
where $\dvol_\Mb$ is the standard volume form on $\Mb=(\Mc,g,\ogth,\tgth)$ 
defined out of its orientation $\ogth$ and of its metric $g$. 
Notice that the above definition of a functional makes use of the usual non-degenerate bilinear pairing 
between $\cc(\Mc)$ and $\func(\Mc)$. It is a well-known result of functional analysis that 
this pairing is non-degenerate. This has two important consequences: First, the map $f\in\cc(\Mc)\mapsto F_f$ 
implicitly defined by \eqref{eqFunctScalar} is injective, thus allowing us to identify the space of functionals 
$\{F_f:\,f\in\cc(\Mc)\}$ with $\cc(\Mc)$. Second, the class of functionals considered so far 
is rich enough to separate off-shell configurations, 
namely, given two different configurations $\phi,\psi\in\func(\Mc)$, 
there exists always $f\in\cc(M)$ such that $F_f(\phi)\neq F_f(\psi)$. 
In fact, this is equivalent to the following statement, which follows from the non-degeneracy 
of the bilinear pairing between $\cc(\Mc)$ and $\func(\Mc)$: 
If $\phi\in\func(\Mc)$ is such that $F_f(\phi)=0$ for all $f\in\cc(\Mc)$, then $\phi=0$. 
Phrased differently, we are saying that, off-shell functionals of the form $F_f$, 
{\em cf.} \eqref{eqFunctScalar}, are faithfully labeled by $f\in\cc(\Mc)$. 
Later we shall restrict functionals to dynamically allowed field configurations. 
This restriction will break the one-to-one correspondence between functionals $F_f$ and $f\in\cc(\Mc)$. 

So far, we did not take into account the dynamics of the real scalar field. 
This is specified by the following partial differential equation: 
\begin{equation}\label{eqScalar}
\Box_\Mb\phi+(m^2+\xi R)\phi=0,
\end{equation}
where $\xi\in\mathbb{R}$, $R$ stands for the scalar curvature built out of $g$, $m^2$ is a real number 
while $\Box_\Mb=g^{ab}\nabla_a\nabla_b:\func(\Mc)\to\func(\Mc)$ is the d'Alembert operator on $\Mb$ 
defined out of the metric $g$ via the associated Levi-Civita connection $\nabla$. Notice that in this paper we are not imposing any constraint on the sign of the mass term since it plays no role. When $\phi\in\func(\Mc)$ is a solution of equation \eqref{eqScalar}, 
we say that $\phi$ is an {\em on-shell} field configuration. 
For convenience, we introduce the differential operator $P=\Box_\Mb+m^2+\xi R$, 
so that \eqref{eqScalar} reduces to $P\phi=0$. We collect all on-shell field configurations in a vector space:
\begin{equation*}
\Sol=\{\phi\in\func(\Mc):\,P\phi=0\}\subset\func(\Mc).
\end{equation*}
It is important to mention that the second order linear differential operator $P$ is formally self-adjoint, {\it cf.} Definition \ref{adjoint}, meaning that, for each $\phi,\psi\in\func(\Mc)$ with supports having compact intersection, one has 
\begin{equation}\label{eqSelfAdjScalar}
\int_\Mc P\phi\;\psi\,\dvol_\Mb=\int_\Mc\phi\;P\psi\,\dvol_\Mb.
\end{equation}
This identity follows from a double integration by parts. Furthermore, $P$ is normally hyperbolic. 
This entails that $P$ admits unique retarded and advanced Green operators $E^+$ and $E^-$, 
see \cite{Bar, BGP, Waldmann}. In particular, it is a Green hyperbolic operator as per Definition \ref{E+E-}. 

Since the functionals $F_f$ are sufficiently many to separate points in $\func(\Mc)$, 
this is also the case for $\Sol$, the latter being a subspace of $\func(\Mc)$. 
We already achieved our first requirement to define classical observables. 
In fact, the functionals $F_f$ can detect any information 
about on-shell field configurations. Specifically, two on-shell configurations $\phi,\psi\in\Sol$ coincide 
if and only if the outcome of their evaluation is the same on all functionals, 
namely $F_f(\phi)=F_f(\psi)$ for all $f\in\cc(\Mc)$. 
Yet, it is the case that some of the functionals considered give no information when evaluated on $\Sol$ 
in the sense that their evaluation on an on-shell configuration always vanishes. 
Here is an explicit example.

\begin{example}\label{exaVanScalar}
Consider $f\in\cc(\Mc)$. Clearly $Pf$ is a smooth function with compact support, 
therefore it makes sense to consider the linear functional $F_{Pf}:\func(\Mc)\to\RR$. 
Just reading out \eqref{eqSelfAdjScalar}, one gets $F_{Pf}(\phi)=F_f(P\phi)$ for all $\phi\in\func(\Mc)$. 
In particular, it follows that $F_{Pf}(\phi)=0$ for all $\phi\in\Sol$, thus $F_{Pf}$ vanishes of $\Sol$. 
\end{example}

The example above shows that functionals of the form $F_f$, after the restriction to $\Sol$, are no longer 
faithfully labeled by $f\in\cc(\Mc)$. In fact, there are indeed redundant functions in $\cc(\Mc)$, 
which provide the same functional on $\Sol$, an example being provided by $0$ and $Pf$, for any $f\in\cc(\Mc)$. 
According to our second requirement, to identify a suitable space of classical observables, 
one has to get rid of such redundancies. 
Therefore, one identifies two functions $f$ and $h$ in $\cc(\Mc)$ if $F_f(\phi)=F_h(\phi)$ for all $\phi\in\Sol$, 
thus restoring a faithful labeling for functionals on solutions. This result can be easily achieved as follows. 
First, one introduces the subspace of those smooth functions with compact support 
providing functionals on $\func(\Mc)$ whose restriction to $\Sol$ vanishes: 
\begin{equation}\label{eqVanScalar}
N=\{f\in\cc(\Mc):\,F_f(\phi)=0,\;\forall\phi\in\Sol\}\subset\cc(\Mc).
\end{equation}
Notice that, according to Example \ref{exaVanScalar}, $P(\cc(\Mc))$, the image of $\cc(\Mc)$ via $P$, 
is a subspace of $N$. \footnote{In Remark \ref{remVanEqualImPScalar} below, 
we shall show that, in the case of the real scalar field, $N=P(\cc(\Mc))$. 
More generally, using the same argument, one can prove an analogous result, 
for any field whose dynamics is ruled by a Green hyperbolic operator.} 
Therefore, we can take the quotient of $\cc(\Mc)$ by $N$, resulting in a new vector space: 
\begin{equation}\label{eqClObsScalar}
\ClObs=\cc(M)/N.
\end{equation}

An equivalence class $[f]\in\ClObs$ yields a functional $F_{[f]}:\Sol\to\RR$ 
specified by $F_{[f]}(\phi)=F_f(\phi)$ for any on-shell configuration $\phi\in\Sol$ 
and for any choice of a representative $f$ of the class $[f]$. 
$F_{[f]}$ is well-defined on account of the definition of $N$ and of the fact that it is evaluated on solutions only. 
Furthermore, by construction, these functionals are in one-to-one correspondence with points in $\ClObs$, 
therefore equivalence classes $[f]\in\ClObs$ faithfully label functionals of the type $F_{[f]}:\Sol\to\RR$. 
Since the quotient by $N$ does not affect the property of separating points in $\Sol$, 
we conclude that $\ClObs$ has the properties required to be interpreted as a space of classical observables, 
namely by evaluation it can distinguish different on-shell configurations 
and, moreover, there are no redundancies since different points provide different functionals on $\Sol$. 
This motivates the fact that we shall refer to $[f]\in\ClObs$ as a {\em classical observable} for the real scalar field. 

\begin{rem}
Besides implementing non-redundancy (essentially by definition), the quotient by $N$, 
corresponds to go on-shell at the level of functionals. Contrary to the functionals in \eqref{eqFunctScalar}, 
which where defined not only on the subspace $\Sol$ of on-shell configurations, but also on all of off-shell fields, 
an equivalence class $[f]\in\ClObs$ provides a functional $F_{[f]}$, which is well-defined only on $\Sol$. 
In fact, for $\phi\in\func(\Mc)\setminus\Sol$, $F_f(\phi)$ depends on the choice of the representative $f\in[f]$. 
Thus $\ClObs$ defines functionals on $\Sol\subset\func(\Mc)$ only. 
In this sense, the quotient by $N$ implements the on-shell condition at the level of functionals. 
\end{rem}

\begin{rem}\label{remVanEqualImPScalar}
Before dealing with the problem of endowing $\ClObs$ with a suitable symplectic structure, 
we would like to point out that $N=P(\cc(\Mc))$. 
In fact, take $f\in\cc(\Mc)$ such that $F_f(\phi)=0$ for all $\phi\in\Sol$. 
According to Theorem \ref{main}, the advanced-minus-retarded operator $E:\ctc(\Mc)\to\func(\Mc)$ 
associated to $P$ maps surjectively onto $\Sol$. Therefore, we can rephrase our condition on $f$ 
as $F_f(Eh)=0$ for all $h\in\ctc(\Mc)$. 
Exploiting \eqref{eqFunctScalar} and recalling Lemma \ref{uniqueGlemma}, one reads 
\begin{equation*}
F_f(Eh)=\int_\Mc f\;Eh\,\dvol_\Mb=-\int_\Mc Ef\;h\,\dvol_\Mb. 
\end{equation*}
According to the hypothesis, the integral on the right-hand-side has to vanish for all $h\in\ctc(\Mc)$, hence $Ef=0$. 
In fact, it would be enough to consider $h\in\cc(\Mc)$ to come to this conclusion. 
Recalling the properties of $E$ again, one finds $f^\prime\in\cc(\Mc)$ such that $Pf^\prime=f$, 
thus showing that $f\in N$ implies $f\in P(\cc(\Mc))$. Since the inclusion $P(\cc(\Mc))\subset N$ 
follows from Example \ref{exaVanScalar}, we conclude that $N=P(\cc(\Mc))$ as claimed. 
\end{rem}

Up to now, a vector space $\ClObs$ providing functionals on $\Sol$ has been determined 
such that points in $\Sol$ can be distinguished by evaluation on these functionals 
and, moreover, $\ClObs$ does not contain redundancies, 
meaning that the map which assigns to $[f]\in\ClObs$ the functional $F_f:\Sol\to\RR$ is injective. 
Yet, to get the classical field theory of the scalar field, a symplectic structure\footnote{Even though 
the term ``symplectic structure'' is mathematically correct, it would be more appropriate 
to refer to this as a constant Poisson structure. Yet, we shall adhere to the common nomenclature 
of quantum field theory on curved spacetimes.} on $\ClObs$ naturally induced by the field equation is still needed. 
For the following construction we shall need the tools developed in Section \ref{sec:2}. 
Let $E^+$ and $E^-$ denote the retarded and advanced Green operators associated to $P=\Box_\Mb+m^2+\xi R$ 
and consider the corresponding advanced-minus-retarded operator $E=E^--E^+$. 
On account of Remark \ref{remVanEqualImPScalar}, we have that $\ClObs=\cc(\Mc)/P(\cc(\Mc))$. 
Therefore, applying Proposition \ref{prpSymplStructure}, we obtain a symplectic structure on $\ClObs$: 
\begin{align}\label{eqSymplScalar}
\tau:\ClObs\times\ClObs\to\RR, && ([f],[h])\mapsto F_f(Eh)=\int_\Mb f\;Eh\,\dvol_\Mb,
\end{align}
where $f$ and $h$ are arbitrary representatives of the equivalence classes $[f]$ and respectively $[h]$ in $\ClObs$. 
The pair $(\ClObs,\tau)$ is the symplectic space of observables 
describing the classical theory of the real scalar field on the globally hyperbolic spacetime $\Mb$ and it is the starting point for the quantization scheme that we shall discuss in the next section. As a preliminary step we discuss some relevant properties. 

\begin{theorem}\label{thmClPropScalar}
Consider a globally hyperbolic spacetime $\Mb=(\Mc,g,\ogth,\tgth)$ 
and let $(\ClObs,\tau)$ be the symplectic space of classical observables defined above for the real scalar field. 
The following properties hold: 
\begin{description}
\item[{\bf Causality}] The symplectic structure vanishes on pairs of observables localized in causally disjoint regions. 
More precisely, let $f,h\in\cc(\Mc)$ be such that $\supp f\cap J_\Mb(\supp h)=\emptyset$. Then $\tau([f],[h])=0$. 
\item[{\bf Time-slice axiom}] Let $\Oc\subset\Mc$ be a globally hyperbolic open neighborhood 
of a spacelike Cauchy surface $\Sigma$ for $\Mb$, namely $\Oc$ is an open neighborhood of $\Sigma$ in $\Mc$ 
containing all causal curves for $\Mb$ whose endpoints lie in $\Oc$. 
In particular, the restriction of $\Mb$ to $\Oc$ provides 
a globally hyperbolic spacetime $\Ob=(\Oc,g\vert_\Oc,\ogth\vert_\Oc,\tgth\vert_\Oc)$. 
Denote with $(\ClObs_\Mb,\tau_\Mb)$ and with $(\ClObs_\Ob,\tau_\Ob)$ the symplectic spaces of observables 
for the real scalar field respectively over $\Mb$ and over $\Ob$. 
Then the map $L:\ClObs_\Ob\to\ClObs_\Mb$ defined by $L[f]=[f]$ for all $f\in\cc(\Oc)$ 
is an isomorphism of symplectic spaces.\footnote{The function in the right-hand-side of the equation 
which defines $L$ is the extension by zero to the whole spacetime of the function appearing in the left-hand-side.}  
\end{description}
\end{theorem}

\begin{proof}
Let us start from the causality property and take $f,h\in\cc(\Mc)$ such that their supports are causally disjoint. 
Recalling the definition of $\tau$ given in \eqref{eqSymplScalar}, one has $\tau([f],[h])=F_f(Eh)$. 
Taking into account the support properties of the advanced-minus-retarded operator $E$, 
one deduces that $\supp(Eh)$ is included in $J_\Mb(\supp h)$, 
which does not intersect the support of $f$ per assumption. 
Since $F_f(Eh)$ is the integral of the pointwise product of $f$ with $Eh$, see \eqref{eqFunctScalar}, 
$\tau([f],[h])=F_f(Eh)=0$ as claimed. 

For the time-slice axiom, consider a globally hyperbolic open neighborhood $\Oc\subset\Mc$ 
of a spacelike Cauchy surface 
for $\Mb$ and consider $\Ob=(\Oc,g\vert_\Oc,\ogth\vert_\Oc,\tgth\vert_\Oc)$, which is a globally hyperbolic spacetime. 
The same construction applied to $\Mb$ and to $\Ob$ provides 
the symplectic spaces $(\ClObs_\Mb,\tau_\Mb)$ and respectively $(\ClObs_\Ob,\tau_\Ob)$. The function
$f\in\cc(\Oc)$ can be extended by zero to the whole $\Mc$ and we denote it still by $f$ with a slight abuse of notation; 
moreover, for each $h\in\cc(\Oc)$, the extension of $Ph=\Box_\Ob h+m^2h$ 
is of the form $Ph=\Box_\Mb h+m^2h$, 
where now $h\in\cc(\Mc)$ denotes the extension of the original $h\in\cc(\Oc)$. 
These observations entail that the map $L:\ClObs_\Ob\to\ClObs_\Mb$ specified by $L[f]=[f]$ for all $f\in\cc(\Oc)$ 
is well-defined. Note that $L$ is linear and that it preserves the symplectic form. 
In fact, given $[f],[h]\in\ClObs_\Ob$, one has 
\begin{equation*}
\tau_\Mb(L[f],L[h]) =\int_\Mc f\;Eh\,\dvol_\Mb=\int_\Oc f\;Eh\,\dvol_\Ob=\tau_\Ob([f],[h]),
\end{equation*}
where the restriction from $\Mc$ to $\Oc$ in the domain of integration is motivated by the fact that, 
per construction, $f=0$ outside $\Oc$. Being a symplectic map, $L$ is automatically injective.
In fact, given $[f]\in\ClObs_\Ob$ such that $L[f]=0$, one has $\tau_\Ob([f],[h])=\tau_\Mb(L[f],L[h])=0$ 
for all $[h]\in\ClObs_\Ob$ and the non-degeneracy of $\tau_\Ob$ entails that $[f]=0$. 
It remains only to check that $L$ is surjective. To this end, starting from any $f\in\cc(\Mc)$, 
we look for $f^\prime\in\cc(\Mc)$ with support inside $\Oc$ such that $[f^\prime]=[f]$ in $\ClObs_\Mb$. 
Recalling that $\Oc$ is an open neighborhood of the spacelike Cauchy surface $\Sigma$ 
and exploiting the usual space-time decomposition of $\Mb$, see Theorem \ref{BS}, 
one finds two spacelike Cauchy surfaces $\Sigma_+,\Sigma_-$ for $\Mb$ included in $\Oc$ 
lying respectively in the future and in the past of $\Sigma$. 
Let $\{\chi^+,\chi^-\}$ be a partition of unity subordinate to the open cover 
$\{I_\Mb^+(\Sigma_-),I_\Mb^-(\Sigma_+)\}$ of $\Mc$. By construction the intersection of the supports 
of $\chi^+$ and of $\chi^-$ is a timelike compact region both of $\Oc$ and of $\Mc$. 
Since $PEf=0$, $\chi^++\chi^-=1$ on $\Mc$ and recalling the support properties of $E$, 
it follows that $f^\prime=P(\chi^-Ef)=-P(\chi^+Ef)$ is a smooth function with compact support inside $\Oc$. 
Furthermore, recalling also the identity $PE^-f=f$, one finds 
\begin{equation*}
\begin{aligned}
f^\prime-f & =P(\chi^-E^-f)-P(\chi^-E^+f)-P(\chi^+E^-f)-P(\chi^-E^-f)\\
& =P(-\chi^-E^+f-\chi^+E^-f).
\end{aligned}
\end{equation*}
The support properties of both the retarded and advanced Green operators $E^+,E^-$ entail that 
$-\chi^-E^+f-\chi^+E^-f$ is a smooth function with compact support on $\Mc$. 
In fact $\supp\chi^\mp\cap\supp(E^\pm f)$ is a closed subset 
of $J_\Mb^\mp(\Sigma_\pm)\cap J_\Mb^\pm(\supp f)$, which is compact. 
This shows that $f^\prime-f\in P(\cc(\Mc))\subset N$, see also Example \ref{exaVanScalar}. 
Therefore we found $[f^\prime\vert_\Oc]\in\ClObs_\Ob$ such that $L[f^\prime\vert_\Oc]=[f]$ 
showing that, besides being injective, the symplectic map $L$ is also surjective 
and hence an isomorphism of symplectic spaces. 
\end{proof}

\begin{rem}
We comment briefly on the apparently different approach, which is often presented in the literature. 
In fact, in place of the pair $(\ClObs,\tau)$, it is quite common 
to consider $\Solsc$, the space of solutions with spacelike compact support, 
endowed with the following symplectic structure: 
\begin{align}\label{eqSymplSolScalar}
\sigma:\Solsc\times\Solsc\to\RR, 
&& (\phi,\psi)\mapsto\int_\Sigma(\phi\nabla_\nb\psi-\psi\nabla_\nb\phi)\,d\Sigma, 
\end{align}
where $\Sigma$ is a spacelike Cauchy surface for the globally hyperbolic spacetime $\Mb=(\Mc,g,\ogth,\tgth)$, 
$\nb$ is the future-pointing unit normal vector field on $\Sigma$, 
and $d\Sigma$ is the induced volume form on $\Sigma$.\footnote{The volume form $d\Sigma$ on $\Sigma$ is defined out of the structure induced on $\Sigma$ itself as a submanifold of the globally hyperbolic spacetime $\Mb$. More explicitly, on $\Sigma$ we take the Riemannian metric $g\vert_\Sigma$ and the orientation specified by the orientation and time-orientation of $\Mb$. Then $d\Sigma$ is the natural volume form defined out of these data.} Notice that the integrand in \eqref{eqSymplSolScalar} is implicitly meant to be restricted to $\Sigma$. Exploiting the fact that only solutions of the field equation are considered, 
one can prove that $\sigma$ does not depend on the choice of the spacelike Cauchy surface $\Sigma$. 
The restriction to the subspace of solutions with spacelike compact support guarantees that the argument of the integral in \eqref{eqSymplSolScalar} is an integrable function. 
We outline below an isomorphism of symplectic spaces between $(\ClObs,\tau)$ and $(\Solsc,\sigma)$: 
\begin{align}
I:\ClObs\to\Solsc, && [f]\mapsto Ef,
\end{align}
where $f\in\cc(\Mc)$ is any representative of $[f]\in\ClObs$ and $E$ denotes the advanced-minus-retarded operator 
associated to the differential operator $P=\Box_\Mb+m^2+\xi R$, which rules the dynamics of the real scalar field. 
The map $I$ is a by-product of Theorem \ref{main} as soon as we remind that, in \eqref{eqClObsScalar}, $N=P(\cc(\Mc))$, 
as shown in Remark \ref{remVanEqualImPScalar}. 
%For each $[f]\in\ClObs$, $I[f]$ is well-defined since two representatives always differ 
%by an element of $N=P(\cc(\Mc))$, see \eqref{eqClObsScalar} and Remark \ref{remVanEqualImPScalar}. 
%Furthermore, $I$ is an isomorphism due to the properties of the advanced-minus-retarded propagator $E$, 
%see Theorem \ref{main}. 
It remains only to check that $\sigma(Ef,Eh)=\tau([f],[h])$ for all $f,h\in\cc(\Mc)$. 
Recalling that $\phi=Ef$ and $\psi=Eh$ are both solutions of the field equation, namely $P\phi=0$ and $P\psi=0$, 
by means of a double integration by parts, one gets the following: 
\begin{equation}\label{eqSymplSolProofScalar}
\begin{aligned}
\int_\Mc f\;Eh\,\dvol_\Mb & 
=\int_{J_\Mb^+(\Sigma)}f\,\psi\,\dvol_\Mb+\int_{J_\Mb^-(\Sigma)}f\,\psi\,\dvol_\Mb\\ 
& =\int_{J_\Mb^+(\Sigma)}(PE^-f)\psi\,\dvol_\Mb+\int_{J_\Mb^-(\Sigma)}(PE^+f)\psi\,\dvol_\Mb\\ 
& =-\int_\Sigma(\nabla_\nb(E^-f))\psi\,d\Sigma+\int_\Sigma(E^-f)\nabla_\nb\psi\,d\Sigma\\ 
& \quad+\int_\Sigma(\nabla_\nb(E^+f))\psi\,d\Sigma-\int_\Sigma(E^+f)\nabla_\nb\psi\,d\Sigma\\ 
& =\int_\Sigma(\phi\nabla_\nb\psi-\psi\nabla_\nb\phi)\,d\Sigma. 
\end{aligned}
\end{equation}
In the first step, we decomposed the integral by splitting the domain of integration into two subsets 
whose intersection has zero measure. 
The second step consisted of exploiting the properties 
of the retarded and advanced Green operators $E^+$ and $E^-$ for $P$. 
Using $E^\mp$ inside the integral over $J_\Mb^\pm(\Sigma)$ allows us to integrate by parts twice. 
For each integral, this operation produces two boundary terms and an integral 
which vanishes since the integrand contains $P\psi=0$. 
Adding together the four boundary terms, one concludes that $\sigma(Ef,Eh)=\tau([f],[h])$ as expected. 
\end{rem}

\subsubsection{Quantum field theory}\label{subsubQuantumScalar}
The next step consists of constructing a quantum field theory for the real scalar field 
out of the classical one, whose content is encoded in the symplectic space $(\ClObs,\tau)$. 
This result is obtained by means of a construction that can be traced back to \cite{Bor62, Haag:1963dh, Uhl62}, while the generalization to curved backgrounds has been discussed from an axiomatic point of view first in \cite{Dimock}. The so-called algebraic approach can be seen as a two-step quantization scheme: In the first one identifies a suitable unital $*$-algebra encoding the structural relations between the observables, such as causality and locality, while, in the second, one selects a state, that is a positive, normalized, linear functional on the algebra which allows us to recover the standard probabilistic interpretation of quantum theories via the GNS theorem. We will focus only on the first step for three different models of free fields, while, for the second we refer to \cite{IV}. We consider the unital $\ast$-algebra $\Ac$ generated over $\CC$ 
by the symbols $\II$ and $\Phi([f])$ for all $[f]\in\ClObs$  
and satisfying the following relations for all $[f],[g]\in\ClObs$ and for all $a,b\in\RR$: 
\begin{align}
\Phi(a[f]+b[g]) & =a\Phi([f])+b\Phi([g]),\label{eqLinearityScalar}\\
\Phi([f])^\ast & =\Phi([f]),\label{eqInvolutionScalar}\\
\Phi([f])\cdot\Phi([h])-\Phi([h])\cdot\Phi([f]) & =i\tau([f],[h])\II.\label{eqCCRScalar}
\end{align}
More concretely, one can start introducing an algebra $A$ consisting of the vector space 
$\bigoplus_{k\in\NN_0}\ClObs_\CC^{\otimes k}$ obtained as the direct sum of all the tensor powers 
of the complexification $\ClObs_\CC$ of the vector space $\ClObs$, 
where we have set $\ClObs_\CC^{\otimes0}=\CC$. 
Therefore, elements of $A$ can be seen as sequences $\{v_k\in\ClObs_\CC^{\otimes k}\}_{k\in\NN_0}$ 
with only finitely many non-zero terms. Each $v_k$ in the sequence is a finite linear combination 
with $\CC$-coefficients of terms of the form $[f_1]\otimes\cdots\otimes[f_k]$ for $[f_1],\dots,[f_k]\in\ClObs$. 
$A$ is endowed with the product $\cdot:A\times A\to A$ specified by 
\begin{align}
\{u_k\}\cdot\{v_k\}=\{w_k\}, && w_k=\sum_{i+j=k}u_i\otimes v_j.
\end{align}
So far, $A$ is an algebra whose generators satisfy \eqref{eqLinearityScalar}. 
We specify an involution $\ast:A\to A$ by setting 
\begin{equation*}
\{\underbrace{0,\dots,0}_{k\mbox{ times}},[f_1]\otimes[f_2]\otimes\cdots\otimes[f_k],0,\dots\}^\ast
=\{\underbrace{0,\dots,0}_{k\mbox{ times}},[f_k]\otimes[f_{k-1}]\otimes\cdots\otimes[f_1],0,\dots\}, 
\end{equation*}
for all $[f_1],\dots,[f_k]\in\ClObs$, and extending it by antilinearity to the whole of $A$. 
Therefore, $A$ is now a $\ast$-algebra implementing the relation \eqref{eqInvolutionScalar} too. 
It is straightforward to realize that the identity of the $\ast$-algebra $A$ is $\II=\{1,0,\dots\}$, 
hence $A$ is also unital. 
Note that an arbitrary element of $A$ can be obtained as a finite $\CC$-linear combination of $\II$ and 
of finite products of elements of the form $\{0,[f],0,\dots\}\in A$, 
which are in one-to-one correspondence with elements of $\ClObs$. 
To match the notation used in the more abstract setting, let us introduce the map $\Phi:\ClObs\to A$, 
$[f]\mapsto\{0,[f],0,\dots\}$, which embeds $\ClObs$ into $A$. 
The $\ast$-algebra $A$ already ``knows'' of the dynamics of the real scalar field 
since this is already encoded in $\ClObs$, 
however, the canonical commutation relations (CCR) \eqref{eqCCRScalar} are still missing. 
Therefore, using the symplectic structure $\tau$ on $\ClObs$, 
we introduce the two-sided $\ast$-ideal $I$ of $A$ generated by terms of the form 
\begin{equation*}
\Phi([f])\cdot\Phi([h])-\Phi([h])\cdot\Phi([f])-i\tau([f],[h])\II,
\end{equation*}
for all $[f],[h]\in\ClObs$. Taking the quotient of $A$ by $I$, one obtains the unital $\ast$-algebra $\Ac=A/I$ 
implementing the canonical commutation relations for the real scalar field. 
Note that, with a slight abuse of notation, we shall denote with $\Phi([f])$ 
also the equivalence class in $\Ac$ of any generator $\Phi([f])$ of $A$, 
thus completely matching the notation used in the more abstract construction of $\Ac$ 
as the unital $\ast$-algebra generated by $\ClObs$ over $\CC$ 
with the relations \eqref{eqLinearityScalar}, \eqref{eqInvolutionScalar}, \eqref{eqCCRScalar}. 
Note in particular that \eqref{eqCCRScalar} is the smeared version of the usual commutation relations. 
This motivates our interpretation of $\Ac$ as the quantum field theory for the real scalar field on $\Mb$. 

\begin{rem}\label{remQuantumScalar}
Before proceeding with the analysis of the properties of the quantum field theory for the real scalar field, 
we would like to emphasize that, under suitable conditions, our quantization procedure perfectly agrees 
with the standard textbook quantization involving creation and annihilation operators. 
In fact, assuming that $\Mb$ is Minkowski spacetime, 
one can relate directly our algebraic approach to the one more commonly used 
by means of an expansion in Fourier modes of the fundamental quantum fields $\Phi([f])$, 
which generate the algebra $\Ac$. 
In particular, one recovers the usual commutation relations between creation and annihilation operators 
out of the canonical commutation relations specified in \eqref{eqCCRScalar} -- see for example \cite{Wald2}. 
This argument should convince the reader that the approach presented above is a very effective extension 
to arbitrary globally hyperbolic spacetimes of the usual quantization procedure for Minkowski spacetime. 
\end{rem}

The properties of the classical field theory presented in Theorem \ref{thmClPropScalar} 
have counterparts at the quantum level as shown by the following theorem. 

\begin{theorem}\label{thmQuantumScalar}
Consider a globally hyperbolic spacetime $\Mb=(\Mc,g,\ogth,\tgth)$ 
and let $\Ac$ be the unital $\ast$-algebra of observables for the real scalar field introduced above. 
The following properties hold: 
\begin{description}
\item[{\bf Causality}] Elements of the algebra $\Ac$ localized in causally disjoint regions commute. 
More precisely, let $f,h\in\cc(\Mc)$ be such that $\supp f\cap J_\Mb(\supp h)=\emptyset$. 
Then $\Phi([f])\cdot\Phi([h])=\Phi([h])\cdot\Phi([f])$. 
\item[{\bf Time-slice axiom}] Let $\Oc\subset\Mc$ be a globally hyperbolic open neighborhood 
of a spacelike Cauchy surface $\Sigma$ for $\Mb$, namely $\Oc$ is an open neighborhood of $\Sigma$ in $\Mc$ 
containing all causal curves for $\Mb$ whose endpoints lie in $\Oc$. 
In particular, the restriction of $\Mb$ to $\Oc$ provides 
a globally hyperbolic spacetime $\Ob=(\Oc,g\vert_\Oc,\ogth\vert_\Oc,\tgth\vert_\Oc)$. 
Denote with $\Ac_\Mb$ and with $\Ac_\Ob$ the unital $\ast$-algebras of observables 
for the real scalar field respectively over $\Mb$ and over $\Ob$. Then the unit-preserving $\ast$-homomorphism 
$\Phi(L):\Ac_\Ob\to\Ac_\Mb$, $\Phi([f])\mapsto\Phi(L[f])$ is an isomorphism of $\ast$-algebras, 
where $L$ denotes the symplectic isomorphism introduced in Theorem \ref{thmClPropScalar}. 
\end{description}
\end{theorem}

\begin{proof}
The quantum version of the causality property follows directly from the classical version 
and the canonical commutation relations. In fact, taking $f,h\in\cc(\Mc)$ with causally disjoint supports, 
one has $\tau([f],[h])=0$ due to Theorem \ref{thmClPropScalar}. 
Recalling the canonical commutation relations \eqref{eqCCRScalar}, 
one has $\Phi([f])\cdot\Phi([h])-\Phi([h])\cdot\Phi([f])=i\tau([f],[h])\II=0$ as claimed. 

Also the time-slice axiom follows directly from its classical counterpart. 
In fact, setting $\Phi(L)\Phi([f])=\Phi(L[f])$ for each generator $\Phi([f])$ of the unital $\ast$-algebra $\Ac_\Ob$ 
uniquely defines a unital $\ast$-homomorphism $\Phi(L):\Ac_\Ob\to\Ac_\Mb$. 
Consider the inverse of $L$, which exists since the classical time-slice axiom states that 
$L$ is a symplectic isomorphism, see Theorem \ref{thmClPropScalar}. The same construction 
applied to $L^{-1}$ provides the unital $\ast$-homomorphism $\Phi(L^{-1}):\Ac_\Mb\to\Ac_\Ob$. 
If $\Phi(L^{-1})$ inverts $\Phi(L)$ on all generators $\Phi([f])$ of $\Ac$, 
then $\Phi(L^{-1})$ is the inverse on $\Phi(L)$ and thus $\Phi(L)$ is a $\ast$-isomorphism. 
Therefore, for all generators $\Phi([f])$ of $\Ac$, we have to check 
the identities $\Phi(L)\Phi(L^{-1})\Phi([f])=\Phi([f])$ and $\Phi(L^{-1})\Phi(L)\Phi([f])=\Phi([f])$. 
But these are obvious consequences of the definitions of $\Phi(L)$ and of $\Phi(L^{-1})$. 
Therefore $\Phi(L)$ is a $\ast$-isomorphism. 
\end{proof}

\subsection{The Dirac field}\label{subDirac}
In this section we present the classical and quantum theory of the Dirac field on a globally hyperbolic spacetime. 
In analogy with the scalar case, we shall discuss first the classical model 
and later develop the corresponding quantum field theory implementing canonical anti-commutation relations. 
Note that, unlike the scalar case, to implement anti-commutation relations, 
we shall need a Hermitian structure in place of a symplectic one. 
Contrary to the real scalar field, the geometry of the space where the Dirac field takes its values 
requires much more attention. This will be the first topic of our presentation, providing the framework 
to write down the Dirac equation. Afterwards, we shall devote some time 
to construct a suitable space of classical observables for on-shell configurations of the Dirac field. 
As in the previous case, we look for a space of sections 
that, by means of integration, provide functionals on on-shell configurations. 
Again, we will be guided by the requirement that the functionals obtained must be able 
to detect any on-shell configuration (separability). A quotient will remove all redundancies 
which might be present in the chosen space of sections. 
Separability and non-redundancy motivate our interpretation of this quotient 
as providing a space of classical observables for the Dirac field. 
To complete the classical part, we shall endow our space of observables with a Hermitian structure, 
which will be used to quantize the classical model, eventually leading to an algebra of observables 
implementing the usual anti-commutation relations for the Dirac field on a globally hyperbolic spacetime. 
Some references discussing the quantum Dirac field on globally hyperbolic spacetimes are 
\cite{Dappiaggi:2009xj, Dimock2, FV02, Sanders, Zahn:2012dz}. 

\subsubsection{Kinematics and dynamics}\label{subsubKynDynDirac}
Contrary to the case of the real scalar field, to specify the natural environment for the Dirac field, 
it is not enough to consider a globally hyperbolic spacetime $\Mb=(\Mc,g,\ogth,\tgth)$. 
In fact, to introduce the kinematics of the Dirac field, one needs more data, namely a spin structure on $\Mb$. 

\begin{definition}
Let $\Mb=(\Mc,g,\ogth,\tgth)$ be an $n$-dimensional globally hyperbolic spacetime. 
Denote with $F\Mb$ the principal $\SO_0(1,n-1)$-bundle of oriented and time-oriented frames on $\Mb$, 
where $\SO_0(1,n-1)$ denotes the component connected to the identity of the proper Lorentz group $\SO(1,n-1)$ in $n$ dimensions. 
Furthermore, consider the spin group $\Spin(1,n-1)$, namely the double cover of $\SO(1,n-1)$ 
and denote with $\Lambda:\Spin(1,n-1)\to\SO(1,n-1)$ the covering group homomorphism. 
Let us also indicate the component connected to the identity of $\Spin(1,n-1)$ with $\Spin_0(1,n-1)$. 
A {\bf spin structure} on $\Mb$ consists of a pair $(S\Mb,\pi)$, 
where the {\bf spin bundle} $S\Mb$ is a principal $\Spin_0(1,n-1)$-bundle, 
and the {\bf spin frame projection} $\pi:S\Mb\to F\Mb$ is a bundle map covering the identity on the base 
and intertwining the right group actions of $\Spin_0(1,n-1)$ on $S\Mb$ and of $\SO_0(1,n-1)$ on $F\Mb$, 
namely such that $\pi(p\,S)=\pi(p)\,\Lambda(S)$ for all $p\in S\Mb$ and for all $S\in\Spin_0(1,n-1)$, 
where the group actions are denoted by juxtaposition. 
\end{definition}

Unfortunately, a spin structure does not always exist on globally hyperbolic spacetimes of arbitrary dimension 
and, even if it does, it might be non-unique. In fact, in general, there are topological obstructions both 
to existence and to uniqueness \cite[Section II.2]{Lawson}. Yet, four-dimensional globally hyperbolic spacetimes, 
the most relevant case to physics, always admit a spin structure, even though this need not be unique. 
First, all spin bundles over a four-dimensional globally hyperbolic spacetime are trivial 
on account of \cite[Section 3]{Isham:1978ec}. Second, all orientable three-manifolds are parallelizable, 
see \cite{Par84}. Since any four-dimensional globally hyperbolic spacetime $\Mb$ can be presented 
as the product of a real line (time) and an oriented 3-manifold (spatial Cauchy surface), see Theorem \ref{BS}, 
it follows that it is parallelizable. In particular, there exists a global section $\epsilon$ of the principal bundle $F\Mb$, 
which consists of an ordered quadruple $(\epsilon_\mu)$ of no-where vanishing orthonormal vector fields on $\Mc$ 
whose orientation is chosen in order to agree with the orientation and the time-orientation of $\Mb$. 
In particular, this entails that $F\Mb$ is trivial, a trivialization being specified by the global frame $\epsilon$ itself. 
In fact, $F\Mb\simeq\Mc\times\SO_0(1,3)$ via the principal bundle map 
$(x,\lambda)\in\Mc\times\SO_0(1,3)\mapsto(\epsilon(x),\lambda)\in F\Mb$. 
Since both $S\Mb$ and $F\Mb$ are trivial for all four-dimensional globally hyperbolic spacetimes $\Mb$, 
it follows that the freedom in the choice of the spin structure actually resides only in that of the spin frame 
projection $\pi:S\Mb\to F\Mb$, which, in turn, reduces to choosing a smooth $\SO_0(1,3)$-valued function over $\Mc$. 
In fact, all possible spin projections between the trivial principal bundles $S\Mb$ and $F\Mb$ are of the form 
\begin{equation*}
\begin{aligned}
\pi:S\Mb\simeq\Mc\times\Spin_0(1,3) & \to F\Mb\simeq\Mc\times\SO_0(1,3),\\ 
(x,S) & \mapsto(x,f(x)\,\Lambda(S)), 
\end{aligned}
\end{equation*}
for some $f\in\func(\Mc,\SO_0(1,3))$. 

Once a spin structure $(S\Mb,\pi)$ has been chosen on the four-dimensional globally hyperbolic spacetime $\Mb$, 
at a kinematic level, a Dirac field is defined to be a section of the vector (Dirac) bundle $D\Mb=S\Mb\times_T\CC^4$ 
with typical fiber $\CC^4$ associated to the principal $\Spin_0(1,3)$-bundle $S\Mb$ via the Dirac representation 
$T=D^{\frac{1}{2},0}\oplus D^{0,\frac{1}{2}}$ of $\Spin_0(1,3)$ on $\CC^4$.\footnote{Note that 
the Dirac representation $T$ is usually regarded as a unitary representation of $\SLL(2,\CC)$ on $\CC^4$, 
yet $\Spin(1,3)$ is isomorphic to $\SLL(2,\CC)$ as a Lie group.} 
Since $S\Mb$ is trivial, all its associated bundles are such, 
thus motivating the more direct definition of the spinor and cospinor bundles given below. 

\begin{definition}
Let $\Mb=(\Mc,g,\ogth,\tgth)$ be a spacetime four-dimensional and globally hyperbolic. 
We define the spinor bundle $D\Mb$ as the trivial vector bundle $M\times\CC^4$, 
while the cospinor bundle $D^\ast\Mb$ is its dual $M\times(\CC^4)^\ast$. 
\end{definition}

At this stage, one can talk about spinors and cospinors as sections of $D\Mb$ and respectively of $D^\ast\Mb$. 
In fact, both bundles being trivial, spinors and cospinors are just smooth functions on $\Mc$ 
taking values in either $\CC^4$ or $(\CC^4)^\ast$. 
Yet, to construct physical quantities out of spinors, such as scalars or currents, 
and to write down the Dirac equation, one still needs $\gamma$-matrices. 

\begin{definition}
Consider the four-dimensional Minkowski space $\MM^4=(\RR^4,\eta)$. 
The Dirac algebra $\Dc$ is the unital algebra 
generated over $\RR$ by an orthonormal basis $\{g_\mu\}_{\mu=0,\dots,3}$ of $\MM^4$ 
and satisfying the relation $g_\mu\,g_\nu+g_\nu\,g_\mu=2\eta_{\mu\nu}\II$. 
\end{definition}

A choice of the $\gamma$-matrices amounts to fixing an irreducible complex representation 
of the Dirac algebra $\Dc$ on the algebra $\Matr(4,\CC)$ of four-by-four complex matrices. 
In fact, any choice of $\gamma_0,\dots,\gamma_3\in\Matr(4,\CC)$ 
satisfying $\gamma_\mu\gamma_\nu+\gamma_\nu\gamma_\mu=2\eta_{\mu\nu}1_4$ 
for all $\mu,\nu=0,1,2,3$ induces an irreducible representation $\rho:\Dc\to\Matr(4,\CC)$ 
defined by $\rho(g_\mu)=\gamma_\mu$. Here $1_4$ denotes the four-by-four identity matrix. 
Note that different choices of the $\gamma$-matrices induce 
equivalent representations \cite{Pau36}, hence the same physical description. 
Yet, to be concrete, we shall consider a specific representation, namely the chiral one. 
Therefore we consider the following family of $\gamma$-matrices: 
\begin{align}
\gamma_0=\begin{pmatrix}
0_2 & 1_2\\
1_2 & 0_2
\end{pmatrix}, &&
\gamma_i=\begin{pmatrix}
0_2 & \sigma_i\\
\sigma_i & 0_2
\end{pmatrix}, i=1,2,3.
\end{align}
where $0_2$, $1_2$ and $\{\sigma_i\}_{i=1,2,3}$ respectively denote 
the zero matrix, the identity matrix and the Pauli matrices in $\Matr(2,\CC)$. 
As one can directly check, the $\gamma$-matrices of our choice satisfy the following relations: 
\begin{equation}\label{eqGammaProperties}
\begin{aligned}
& \gamma_\mu\gamma_\nu+\gamma_\nu\gamma_\mu=2\eta_{\mu\nu}1_4,\;\;\mu,\nu=0,\dots,3,\\ 
& \begin{aligned}
& \gamma_0^\dag=\gamma_0, && \gamma_i^\dag=-\gamma_i,\;\;i=1,2,3,\\ 
& %\gamma_\mu=\gamma_0\gamma_\mu\gamma_0^{-1}, &&  
\ol{\gamma_\mu}=-\gamma_2\gamma_\mu\gamma_2^{-1}%\\ 
%& xxx\ol{\gamma_2}\gamma_2=1_4, 
&& \gamma_0\rho(n)>0, 
\end{aligned}
\end{aligned}
\end{equation}
$n$ is any future pointing timelike vector in $\MM^4$, 
$\ol{(\cdot)}$ denotes the complex conjugation of each entry, $(\cdot)^T$ is the transpose of a matrix 
and $(\cdot)^\dag=\overline{(\cdot)^T}$. Since we defined the spinor bundle $D\Mb$ 
as a trivial bundle over $\Mc$ with fiber $\CC^4$, 
we can easily interpret the $\gamma$-matrices as endomorphisms of this bundle:
\begin{align}
\gamma_\mu:D\Mb & \to D\Mb, & (x,\sigma) & \mapsto(x,\gamma_\mu\sigma).
\end{align}
Note that the action of the $\gamma$-matrices on cospinors is obtained 
by composing $\gamma_\mu:D\Mb\to D\Mb$ on the right. 
In fact, one can read the cospinors bundle $D^\ast\Mb$ as a bundle 
whose fibers are $\CC$-linear functionals on the corresponding fiber of $D\Mb$. 
Therefore it is natural to express the action of $\gamma_\mu$ on $D^\ast\Mb$ as 
$(x,\omega)\in D^\ast\Mb\mapsto (x,\omega)\circ\gamma_\mu=(x,\omega\circ\gamma_\mu)\in D^\ast\Mb$. 
For simplicity, in the following the composition will be left understood. 
Let us also mention that, the $\gamma$-matrices being invertible, see \eqref{eqGammaProperties}, 
the induced vector bundle maps are actually isomorphisms. 
In particular, by means of the $\gamma$-matrices, one can introduce complex anti-linear vector bundle isomorphisms 
covering the identity which implement adjunction and charge conjugation: 
\begin{align}
A:D\Mb & \to D^\ast\Mb, & (x,\sigma) & \mapsto(x,\sigma^\dag\gamma_0),\label{eqAdjSDirac}\\
\Cs:D\Mb & \to D\Mb, & (x,\sigma) & \mapsto(x,\ol{\gamma_2\sigma}),\label{eqChConjSDirac}\\
\Cc:D^\ast\Mb & \to D^\ast\Mb, & (x,\omega) & \mapsto(x,\ol{\omega}\gamma_2),\label{eqChConjCDirac}
\end{align}
From \eqref{eqGammaProperties} one can show that $A$ intertwines $\Cs$ and $\Cc$ up to a minus sign, 
namely $A\circ\Cs=-\Cc\circ A$. 

Furthermore, let us fix an oriented, orthochronous, orthonormal co-frame $e=(e^\mu)_{\mu=0,\dots,3}$ on $\Mb$ 
once and for all. The $e^\mu$'s are no-where vanishing one-forms on $\Mc$ 
which allow to completely reconstruct the structure of the globally hyperbolic spacetime $\Mb=(\Mc,g,\ogth,\tgth)$: 
$g=\eta_{\mu\nu}e^\mu\otimes e^\nu$, $\ogth=[e^0\wedge\cdots\wedge e^{3}]$ and $\tgth=[e^0]$, 
where the square brackets are used to indicate the (time-)orientation induced by the enclosed form. 
Fixing $e$ is completely equivalent to the choice of a frame $\epsilon=(\epsilon_\mu)$, 
namely a section of the frame bundle $F\Mb$. In fact, $\epsilon$ can be obtained from $e$ 
setting $\epsilon_\mu=\eta_{\mu\nu}(e^\nu)^\flat$ and similarly $e$ can be obtained from $\epsilon$ 
as $e^\mu=\eta^{\mu\nu}\epsilon_\nu^\sharp$, where $(\cdot)^\flat$ and $(\cdot)^\sharp$ 
are the canonical $g$-induced isomorphism which lower and raise the indices of tensors on $\Mc$, 
while $\eta$ denotes the metric of Minkowski space $\MM^4$. 
Using the fixed co-frame $e$ of $\Mb$, one can specify a one-form $\gamma$ over $\Mc$ 
taking values in the bundle of endomorphisms of the spinor bundle $D\Mb$: 
\begin{align}
\gamma:T\Mc\to\Endo(D\Mb), && v\mapsto e^\mu(v)\gamma_\mu.
\end{align}
Note that $\{e^\mu(v)\in\RR\}_{\mu=0,\dots,3}$ are the components of $v\in T_xM$ 
with respect to the frame $\epsilon$ obtained raising the indices of the fixed co-frame $e$, 
namely $v=e^\mu(v)\epsilon_\mu$. 

To write down the Dirac equation, the last necessary ingredient is a suitable covariant derivative 
on the spinor and cospinor bundles. Abstractly, one could start from the Levi-Civita connection, 
which is a principal bundle connection on the frame bundle $F\Mb$. 
Exploiting the spin structure $(S\Mb,\pi$), one can pull-back the Levi-Civita connection form along $\pi$ 
and then lift it along the double cover $\Lambda:\Spin_0(1,3)\to\SO_0(1,3)$ 
to obtain a $1$-form on $S\Mb$ taking values in $\mathfrak{so}(1,3)$ the Lie algebra of both $\SO_0(1,3)$ and $\Spin_0(1,3)$. 
This procedure actually provides a principal bundle connection on $S\Mb$. 
Thinking of the spinor bundle as a vector bundle associated to $S\Mb$, 
a covariant derivative is naturally induced from the connection on $S\Mb$. 
This covariant derivative is the one relevant to the Dirac field. 
Yet, motivated by the fact that $S\Mb$ is trivial (and hence so is any associated bundle), 
we preferred to define directly the spinor bundle as a certain trivial vector bundle. 
Following this approach, it seems more appropriate to define the covariant derivative on $D\Mb$ explicitly: 
\begin{align}
\nabla:\sect(T\Mc)\otimes\sect(D\Mb)\to\sect(D\Mb), && (X,\sigma)\mapsto \nabla_X\sigma
=\partial_X\sigma+\frac{1}{4}X^\mu\ChS_{\mu\nu}^\rho\gamma_\rho\gamma^\nu\sigma,
\end{align}
where $\sigma$ is regarded as a smooth $\CC^4$-valued function on $\Mc$, 
$X^\mu=e^\mu(X)$ are the components of $X$ in the fixed frame 
and $\ChS_{\mu\nu}^\rho=e^\rho(\nabla_{\epsilon_\mu}\epsilon_\nu)$ 
are the Christoffel symbols of the Levi-Civita connection with respect to the given frame. 
The covariant derivative is naturally extended to cospinors by imposing the identity 
\begin{equation*}
\partial_X(\omega(\sigma))=(\nabla_X\omega)(\sigma)+\omega(\nabla_X\sigma), 
\end{equation*}
for each vector field $X\in\sect(T\Mc)$, for each spinor field $\sigma\in\sect(D\Mb)$ 
and for each cospinor $\omega\in\sect(D\Mb)$. 
We extend further $\nabla$ to mixed spinor-tensor fields via the Leibniz rule. 
As an example, we show that $\nabla\gamma=0$, the computation being carried out using frame components: 
\begin{equation}\label{eqNablaGamma}
\begin{aligned}
\nabla_{\epsilon_\mu}\gamma & =-\ChS_{\mu\nu}^\rho e^\nu\otimes\gamma_\rho
+\frac{1}{4}\ChS_{\mu\nu}^\sigma e^\rho\otimes[\gamma_\sigma\gamma^\nu,\gamma_\rho]\\
& =-\ChS_{\mu\nu}^\rho e^\nu\otimes\gamma_\rho+\frac{1}{4}\ChS_{\mu\nu}^\sigma e^\rho
\otimes(\gamma_\sigma\{\gamma^\nu,\gamma_\rho\}-\{\gamma_\sigma,\gamma_\rho\}\gamma^\nu)\\
& =-\frac{1}{2}\ChS_{\mu\nu}^\rho e^\nu\otimes\gamma_\rho
-\frac{1}{2}\ChS_{\mu\tau}^\sigma\eta_{\sigma\nu}\eta^{\tau\rho}e^\nu\otimes\gamma_\rho=0. 
\end{aligned}
\end{equation}
To conclude, we exploited the identity 
$\ChS_{\mu\nu}^\sigma\eta_{\rho\sigma}+\ChS_{\mu\rho}^\sigma\eta_{\nu\sigma}=0$, 
which follows from $\nabla g=0$ written in frame components. 
Notice that $[\cdot,\cdot]$ and $\{\cdot,\cdot\}$ are used here 
to denote respectively the commutator and the anti-commutator of matrices. 

Using the covariant derivatives both for spinors and for cospinors, 
together with our choice of the $\gamma$-matrices, 
we can introduce the first order linear differential operators 
$\nss:\sect(D\Mb)\to\sect(D\Mb)$ and $\nsc:\sect(D^\ast\Mb)\to\sect(D^\ast\Mb)$ defined according to 
\begin{align}
\nss\sigma=\Tr_g(\gamma\,\nabla\sigma), && \forall\sigma\in\sect(D\Mb),\label{eqNablaSlashS}\\
\nsc\omega=\Tr_g(\nabla\omega\,\gamma), && \forall\omega\in\sect(D^\ast\Mb),\label{eqNablaSlashC}
\end{align}
where $\Tr_g$ denotes the metric-contraction of the covariant two-tensor $\gamma\,\nabla\sigma$ taking values 
in $D\Mb$ and similarly for $\nabla\omega\,\gamma\in\sect(T^\ast\Mc\otimes T^\ast\Mc\otimes D^\ast\Mb)$. 
With respect to the fixed frame $(\epsilon_\mu)_{\mu=0,\dots,3}$ 
\eqref{eqNablaSlashS} reads $\nss\sigma=\eta^{\mu\nu}\gamma_\mu\nabla_{\epsilon_\nu}$ 
while, using the abstract tensor notation, one has $\nss\sigma=g^{ab}\gamma_a\nabla_b\sigma$. 
Similar considerations apply to $\nsc$. 

We can now write down the Dirac equation both for spinors and for cospinors in the usual form: 
\begin{align}\label{eqDirac}
i\nss\sigma-m\sigma=0, && -i\nsc\omega-m\omega=0.
\end{align}
For convenience, we introduce the differential operators $\Ps=i\nss-m\,\id_{\sect(D\Mb)}$ for spinors 
and $\Pc=-i\nss-m\,\id_{\sect(D^\ast\Mb)}$ for cospinors. 
Exploiting the properties of the $\gamma$-matrices listed in \eqref{eqGammaProperties}, 
and taking into account the action of the adjunction \eqref{eqAdjSDirac} 
and of the charge conjugations \eqref{eqChConjSDirac} and \eqref{eqChConjCDirac} on sections, 
one can easily prove that $A\circ\Ps=\Pc\circ A$, $\Cs\circ\Ps=\Ps\circ\Cs$ and $\Cc\circ\Pc=\Pc\circ\Cc$. 

To investigate the properties of the Dirac equation, we introduce an integral pairing between sections 
of the spinor bundle $D\Mb$ and sections of its dual $D^\ast\Mb$. 
For each pair of sections $\omega\in\sect(D^\ast\Mb)$ and $\sigma\in\sect(D\Mb)$ 
such that $\supp\omega\cap\supp\sigma$ is compact, we define 
\begin{equation*}
\lan\omega,\sigma\ra=\int_\Mc\omega(\sigma)\,\dvol_\Mb.
\end{equation*}
This pairing is per construction linear in both arguments. Since the adjunction $A$ maps spinors to cospinors, 
we can use it to form integral pairings between spinors and between cospinors: 
\begin{align}\label{eqInnerProdDirac}
\ips{\sigma}{\tau}=\lan A\sigma,\tau\ra, && \ipc{\omega}{\zeta}=\lan\zeta,A^{-1}\omega\ra,
\end{align}
where $\sigma,\tau\in\sect(D\Mb)$ are such that the intersection of their supports is compact 
and $\omega,\zeta\in\sect(D^\ast\Mb)$ satisfy the same condition. 
Notice that, due to the anti-linearity of $A$, both pairings defined above are 
linear in the second argument and anti-linear in the first. 
Furthermore, it is easy to check that $\ips{\cdot}{\cdot}$ induces a Hermitian form on $\sectc(D\Mb)$. 
In fact, given $\sigma,\tau\in\sect(D\Mb)$ such that their supports have compact overlap, 
one has $\sigma^\dag(\gamma_0\tau)=\tau^T(\gamma_0^T\overline{\sigma})$, 
hence, using also the identity $\gamma_0^\dag=\gamma_0$, one deduces that 
\begin{equation}\label{eqHermInnProdDirac}
\overline{\ips{\sigma}{\tau}}=\overline{\int_\Mc\sigma^\dag(\gamma_0\tau)\,\dvol_\Mb}
=\int_\Mc\tau^\dag(\gamma_0\sigma)\,\dvol_\Mb=\ips{\tau}{\sigma}.
\end{equation}
Similarly, $\ipc{\cdot}{\cdot}$ induces a Hermitian form on $\sectc(D^\ast\Mb)$. 

Using the properties \eqref{eqGammaProperties}, one realizes that 
$\gamma:\sect(D\Mb)\to\sect(D\Mb)$ is formally self-adjoint with respect to $\ips{\cdot}{\cdot}$. 
Similarly, $\gamma:\sect(D^\ast\Mb)\to\sect(D^\ast\Mb)$ has the same property 
with respect to $\ipc{\cdot}{\cdot}$. 
Furthermore, both $\nss$ and $\nsc$ coincide with their formal adjoints with respect to $\ips{\cdot}{\cdot}$ 
and respectively to $\ipc{\cdot}{\cdot}$ up to the sign and, moreover, still up to the sign, 
they are formal duals of each other with respect to $\lan\cdot,\cdot\ra$. 
Specifically, consider any pair of spinors $\sigma,\tau\in\sect(D\Mb)$, 
any pair of cospinors $\omega,\zeta\in\sect(D^\ast\Mb)$ 
and any pair formed by a spinor $\upsilon\in\sect(D\Mb)$ and a cospinor $\varpi\in\sect(D^\ast\Mb)$. 
Assume that the supports of the sections in each pair have compact overlap. Then the following identities hold: 
\begin{align}\label{eqNablaSlashAntiAdj}
\ips{\nss\sigma}{\tau}=-\ips{\sigma}{\nss\tau}, && \ipc{\nsc\omega}{\zeta}=-\ipc{\omega}{\nsc\zeta}, 
&& \lan\nsc\varpi,\upsilon\ra=-\lan\varpi,\nss\upsilon\ra.
\end{align}
For the sake of clarity, below we prove the first identity. The proof of the others is analogous. 
\begin{equation}\label{eqNablaSlashAntiAdjProof}
\begin{aligned}
\ips{\nss\sigma}{\tau}+\ips{\sigma}{\nss\tau}
& =\int_\Mc\Tr_g\Big(\big(A(\gamma\,\nabla\sigma)\big)(\tau)+(A\sigma)(\gamma\,\nabla\tau)\Big)\,\dvol_\Mb\\
& =\int_\Mc\Tr_g\Big(\big(\nabla(A\sigma)\big)(\gamma\tau)+(A\sigma)\big(\nabla(\gamma\tau)\big)\Big)
\,\dvol_\Mb\\
& =\int_\Mc\Tr_g\Big(\nabla\big((A\sigma)(\gamma\tau)\big)\Big)\,\dvol_\Mb\\
& =\int_\Mc\dd\ast\big((A\sigma)(\gamma\tau)\big)=0,
\end{aligned}
\end{equation}
where we used \eqref{eqGammaProperties}, the Leibniz rule, 
the identity $\nabla\gamma=0$ proved in \eqref{eqNablaGamma} and Stokes' theorem. 
We remind the reader that $\dd$ is the exterior derivative for differential forms over $\Mc$, 
while $\ast$ denotes the Hodge star operator defined out of the metric $g$ 
and out of the orientation $\ogth$ of the globally hyperbolic spacetime $\Mb$. 
From the identities in \eqref{eqNablaSlashAntiAdj} it follows that 
both $\Ps$ and $\Pc$ are formally self-adjoint differential operators. 
Therefore, it is enough to exhibit retarded and advanced Green operators for each of them 
to conclude that those are unique, see Lemma \ref{uniqueGlemma}, 
and that $\Ps$ and $\Pc$ are Green hyperbolic, see Definition \ref{E+E-}. 
Furthermore, Proposition \ref{uniqueG} entails that the retarded/advanced Green operator for $\Pc$ 
is the formal dual of the advanced/retarded Green operator for $\Ps$. 
To construct the Green operators we are interested in, we observe that 
$\nss^2=\nss\circ\nss$ and $\nsc^2=\nsc\circ\nsc$ are both normally hyperbolic operators. 
Consider for example $\nss^2$. For each $\sigma\in\sect(D\Mb)$, one has the following: 
\begin{equation}\label{eqNablaSlashSquared}
\begin{aligned}
\nss^2\sigma
& =\gamma^\mu\nabla_{\epsilon_\mu}(\gamma^\nu\nabla_{\epsilon_\nu}\sigma)%\\
=\gamma^\mu\gamma^\nu\nabla_{\epsilon_\mu}\nabla_{\epsilon_\nu}\sigma\\
& =-\gamma^\nu\gamma^\mu\nabla_{\epsilon_\mu}\nabla_{\epsilon_\nu}\sigma
+2\eta^{\mu\nu}\nabla_{\epsilon_\mu}\nabla_{\epsilon_\nu}\sigma\\
& =-\gamma^\nu\gamma^\mu(\nabla_{\epsilon_\nu}\nabla_{\epsilon_\mu}\sigma
+\nabla_{\epsilon_\mu}\nabla_{\epsilon_\nu}\sigma-\nabla_{\epsilon_\nu}\nabla_{\epsilon_\mu}\sigma)
+2\eta^{\mu\nu}\nabla_{\epsilon_\mu}\nabla_{\epsilon_\nu}\sigma\\
& =-\nss^2\sigma+2\Box_\nabla\sigma+\frac{1}{2}R\sigma.
\end{aligned}
\end{equation}
Notice that the last computation has been performed with respect to the chosen co-frame 
$e=(e^\mu)_{\mu=0,\dots,3}$. $\Box_\nabla:\sect(D\Mb)\to\sect(D\Mb)$ denotes 
the d'Alembert operator constructed out of the connection $\nabla$, while 
%$R_{\mu\nu\rho}{}^\sigma$ denote the components of the Riemann tensor with respect to $\epsilon$, 
%namely $R(\epsilon_\mu,\epsilon_\nu)\epsilon_\rho=R_{\mu\nu\rho}{}^\sigma\epsilon_\sigma$ and 
$R$ is the scalar curvature. Let us also mention that, for the last equality in the computation above, 
we used the first Bianchi identity and the anti-commutation relations between the $\gamma$-matrices. 
Therefore, one concludes that $\nss^2=\Box_\nabla+R/4$, 
hence it is normally hyperbolic, {\em cf.} Definition \ref{normhyp}. 

\begin{proposition}\label{prpGreenHypDirac}
Let $\Mb=(\Mc,g,\ogth,\tgth)$ be a four-dimensional globally hyperbolic spacetime 
together with a co-frame $e=(e^\mu)_{\mu=0,\dots,3}$. 
The first order linear differential operators $\Ps:\sect(D\Mb)\to\sect(D\Mb)$ and 
$\Pc:\sect(D^\ast\Mb)\to\sect(D^\ast\Mb)$, which rule the dynamics of spinors and respectively of cospinors, 
are formally self-adjoint with respect to $\ips{\cdot}{\cdot}$ and respectively to $\ipc{\cdot}{\cdot}$. 
Furthermore, they are both Green hyperbolic. In particular, their retarded and advanced Green operators are given by 
\begin{align}\label{eqGreenOpDirac}
\Es^\pm=\Ps\Fs^\pm, && \Ec^\pm=\Pc\Fc^\pm,
\end{align}
where $\Fs^+$ and $\Fs^-$ denote the retarded and advanced Green operators 
for the Green hyperbolic operator $\Ps^2=\Ps\Ps:\sect(D\Mb)\to\sect(D\Mb)$, 
while $\Fc^+$ and $\Fc^-$ denote those corresponding 
to the Green hyperbolic operator operator $\Pc^2=\Pc\Pc:\sect(D^\ast\Mb)\to\sect(D^\ast\Mb)$.
\end{proposition}

\begin{proof}
Formal self-adjointness of $\Ps$ follows directly from \eqref{eqNablaSlashAntiAdj}. 
In fact, the minus sign which appears while integrating by parts $\nss$ is reabsorbed 
by the imaginary unit while passing from one argument of $\ips{\cdot}{\cdot}$ to he other 
due to anti-linearity in the first argument of the pairing. 
A similar argument shows that also $\Pc$ is formally self-adjoint with respect to $\ipc{\cdot}{\cdot}$. 

It is enough to exhibit retarded and advanced Green operators 
to conclude that both $\Ps$ and $\Pc$ are Green hyperbolic, {\em cf.} Definition \ref{E+E-}. 
Specifically, in the following we shall prove that the operators introduced in \eqref{eqGreenOpDirac} 
are actually the sought Green operators. We focus on the case of spinors, the other being completely analogous. 
First of all, we prove that the formally self-adjoint operator $\Ps^2$ is Green hyperbolic as claimed. 
In fact, on account of the identity $\nss^2=\Box_\nabla+R/4$, 
which is a consequence of \eqref{eqNablaSlashSquared}, one concludes that $\Ps^2=-\nss^2-2im\nss+m^2$. 
Therefore, according to Definition \ref{normhyp}, $-\Ps^2$ is normally hyperbolic, 
hence it admits retarded and advanced Green operators, 
see \cite{Bar} and \cite[Chapter 3]{BGP}. Reversing the sign, one gets retarded and advanced Green operators 
$\Fs^+$ and $\Fs^-$ for $\Ps^2$, thus showing that $\Ps^2$ is Green hyperbolic. 
To conclude the proof, we show that $\Es^+=\Ps\Fs^+$ and $\Es^-=\Ps\Fs^-$ 
are retarded and advanced Green operators for $\Ps$. 
The support properties of retarded and advanced Green operators are satisfied 
since $\Fs^+$ and $\Fs^-$ are retarded and advanced Green operators 
and, moreover, being a differential operator, $\Ps$ does not enlarge supports. 
Indeed, for each $\sigma \in\ Gamma_{pc/fc}(D \Mb)$, one has $\Ps\Es^\pm\sigma=\Ps^2\Fs^\pm\sigma=\sigma$. 
It remains only to check that $\Es^\pm\Ps\sigma=\sigma$. 
Let us take $\tau\in\sectc(D\Mb)$ and consider $\ips{\Es^\pm\Ps\sigma}{\tau}$: 
\begin{equation*}
\begin{aligned}
\ips{\Es^\pm\Ps\sigma}{\tau}& =\ips{\Es^\pm\Ps\sigma}{\Ps\Es^\mp\tau}
=\ips{\Ps\Es^\pm\Ps\sigma}{\Es^\mp\tau}\\
& =\ips{\Ps\sigma}{\Es^\mp\tau}=\ips{\sigma}{\Ps\Es^\mp\tau}=\ips{\sigma}{\tau}.
\end{aligned}
\end{equation*}
In the last chain of identities we exploited repeatedly the formal self-adjointness of $\Ps$ 
and the identity $\Es^\pm\Ps\upsilon=\upsilon$, which holds true for all $\upsilon \in \Gamma_{pc/fc}(D \Mb)$. 
Since $\ips{\cdot}{\cdot}$ provides a non-degenerate pairing between $\sect(D\Mb)$ and $\sectc(D\Mb)$, 
we deduce that $\Es^\pm\Ps\sigma=\sigma$, thus completing the proof. 
\end{proof}

Indeed, the fact that $\Ps$ and $\Pc$ are formally self-adjoint with respect to $\ips{\cdot}{\cdot}$ 
and to $\ipc{\cdot}{\cdot}$ has a counterpart involving the corresponding retarded and 
advanced Green operators on account of Lemma \ref{uniqueGlemma}. 
A similar argument applies to the fact $\Pc$ is the formal dual of $\Ps$ with respect to $\lan\cdot,\cdot\ra$, 
see \eqref{eqNablaSlashAntiAdj} and Proposition \ref{uniqueG}. Summing up, one has the following identities 
for all $\sigma,\tau\in\sectc(D\Mb)$ and for all $\omega,\zeta\in\sectc(D^\ast\Mb)$: 
\begin{align}
\ips{\Es^\pm\sigma}{\tau}=\ips{\sigma}{\Es^\mp\tau}, 
&& \ipc{\Ec^\pm\omega}{\zeta}=\ipc{\omega}{\Ec^\mp\zeta}, 
&& \lan\Ec^\pm\omega,\sigma\ra=\lan\omega,\Es^\mp\sigma\ra.
\end{align}
Proposition \ref{prpGreenHypDirac} concludes our discussion about the dynamics of the Dirac field. 
In fact, introducing the advanced-minus-retarded operators $\Es=\Es^--\Es^+$ and $\Ec=\Ec^--\Ec^+$ 
corresponding to $\Ps$ and respectively to $\Pc$, one can easily represent all on-shell spinors and cospinors 
over the four-dimensional globally hyperbolic spacetime $\Mb$, see Theorem \ref{main}. 

\subsubsection{Classical observables}
From the previous section we know that the on-shell configurations of the Dirac field are 
either spinors or cospinors, namely sections of either $D\Mb$ or $D^\ast\Mb$, 
satisfying the Dirac equation \eqref{eqDirac}. 
We shall consider now a class of functionals on these field configurations.  
As further properties, we shall require that this class is large enough to separate different on-shell configurations 
and that its elements are represented faithfully by some vector space, to be endowed later 
with the Hermitian structure canonically induced by the Dirac Lagrangian. 
Let us start with $\tau\in\sectc(D\Mb)$ and $\zeta\in\sectc(D^\ast\Mb)$ 
to introduce the functional $S_\tau$ for spinors and the functional $C_\omega$ for cospinors: 
\begin{equation}\label{eqFunctDirac}
\begin{aligned}
S_\tau:\sect(D\Mb) & \to\CC, & \sigma & \mapsto\ips{\tau}{\sigma},\\
C_\zeta:\sect(D^\ast\Mb) & \to\CC, & \omega & \mapsto\ipc{\zeta}{\omega}.
\end{aligned}
\end{equation}
Since both $\ips{\cdot}{\cdot}$ and $\ipc{\cdot}{\cdot}$ induce non-degenerate bilinear pairings 
on $\sectc(D\Mb)\times\sect(D\Mb)$ and respectively on $\Gamma_0(D^\ast\Mb)\times\Gamma(D\Mb)$, 
one can identify the vector spaces of functionals $\{S_\tau:\tau\in\sectc(D\Mb)\}$ 
and $\{C_\zeta:\zeta\in\sectc(D^\ast\Mb)\}$ with $\sectc(D\Mb)$ and respectively with $\sectc(D^\ast\Mb)$. 
These identifications are implemented via the anti-linear maps 
$\tau\in\sectc(D\Mb)\mapsto S_\tau$ and $\zeta\in\sectc(D^\ast\Mb)\mapsto C_\zeta$. 
Let us stress one fact, which follows from non-degeneracy 
of the pairings $\ips{\cdot}{\cdot}$ and $\ipc{\cdot}{\cdot}$. 
The functionals $\{S_\tau:\tau\in\sectc(D\Mb)\}$ on spinors 
and the functionals $\{C_\zeta:\zeta\in\sectc(D^\ast\Mb)\}$ on cospinors are sufficiently many 
to separate different off-shell field configurations, hence on-shell ones in particular. 
Therefore, our separability requirement is already achieved. 

The functionals introduced above do not take into account the dynamics for Dirac fields. 
We can easily overcome this hurdle restricting the domains to on-shell configurations. 
Let us introduce the spaces of on-shell spinors and of on-shell cospinors: 
\begin{align}
\sSol=\{\sigma\in\sect(D\Mb):\,\Ps\sigma=0\}, && \cSol=\{\omega\in\sect(D^\ast\Mb):\,\Pc\omega=0\}. 
\end{align}
Given $\tau\in\sectc(D\Mb)$ and $\zeta\in\sectc(D^\ast\Mb)$, with a slight abuse of notation, 
we denote the restrictions $S_\tau:\sSol\to\CC$ and $C_\zeta:\cSol\to\CC$ of the original functionals 
introduced in \eqref{eqFunctDirac} by the same symbols. 
This restriction causes some redundancies in the spaces $\sectc(D\Mb)$ and $\sectc(D^\ast\Mb)$, 
which do not faithfully represent the functionals after the restriction to on-shell configurations. 
This fact is explicitly shown in the next example. 

\begin{example}\label{exaTrivDirac}
Let us consider $\tau\in\sectc(D\Mb)$. Since $\Ps$ does not enlarge supports, 
$\Ps\tau$ is still a compactly supported section of $D\Mb$. Hence, we can consider the functional $S_{\Ps\tau}$. 
We show that this functional vanishes when restricted to $\sSol$. 
In fact, according to Proposition \ref{prpGreenHypDirac}, one deduces that 
$\Ps$ is formally self-adjoint with respect to $\ips{\cdot}{\cdot}$, which entails that $S_{\Ps\tau}(\sigma)=S_\tau(\Ps\sigma)=0$ for all $\sigma\in\sSol$. 
In full analogy, $C_{\Pc\zeta}$ vanishes on $\cSol$ for all $\zeta\in\sectc(D^\ast\Mb)$ 
since $\Pc$ is formally self-adjoint with respect to $\ipc{\cdot}{\cdot}$. 
Summing up, the elements of $\Ps(\sectc(D\Mb))$ and of $\Pc(\sectc(D^\ast\Mb))$ are redundant 
since they provide only trivial functionals respectively on $\sSol$ and on $\cSol$. 
\end{example}

To implement our second requirement for classical observables, 
namely that the space representing functionals should be free of redundancies 
(or, equivalently, functionals should be represented faithfully by this space), 
we simply take a quotient by the subspace of (co)spinors inducing functionals which vanish on-shell: 
\begin{align}
\sVan & =\{\tau\in\sectc(D\Mb):\,S_\tau(\omega)=0,\;\forall\sigma\in\sSol\}\\
\cVan & =\{\zeta\in\sectc(D^\ast\Mb):\,C_\zeta(\omega)=0,\;\forall\omega\in\cSol\},.
\end{align}
As anticipated, we introduce the quotient spaces 
\begin{align}\label{eqClObsDirac}
\sClObs=\sectc(D\Mb)/\sVan, && \cClObs=\sectc(D^\ast\Mb)/\cVan.
\end{align}
$\sClObs$ and $\cClObs$ are regarded as the spaces of linear classical observables 
respectively for spinors and for cospinors. In fact, these spaces faithfully represent 
the restrictions to on-shell configurations of the functionals defined in \eqref{eqFunctDirac}, 
which are sufficiently many to distinguish between different on-shell configurations. 
For example, consider the case of spinors. It is clear that the equivalence class $[\tau]\in\sClObs$ 
induces a unique functional $S_\tau:\sSol\to\CC$, independent of the choice of $\tau\in[\tau]$. 
Indeed different representatives induce different functionals on $\sect(D\Mb)$ (off-shell), 
but, per definition of $\sVan$, all these functionals have the same restriction to $\sSol$ (on-shell). 
Therefore one has an anti-linear map $[\tau]\in\sClObs\mapsto S_\tau$. 
Again per definition of $\sVan$, this map is injective, thus providing a faithful way to represent by means 
of $\sClObs$ the restrictions to on-shell configurations of the functionals in \eqref{eqFunctDirac}. 

\begin{rem}
Using our knowledge about the dynamics of the Dirac field, {\em cf.} Section \ref{subsubKynDynDirac}, 
we can prove that 
\begin{align}\label{eqVanImPDirac}
\sVan=\Ps(\sectc(D\Mb)), && \cVan=\Pc(\sectc(D^\ast\Mb)),
\end{align}
meaning that all redundant functionals are of the form presented in Example \ref{exaTrivDirac}. 
As always, we focus our attention to the case of spinors only, 
the argument being basically the same in the case of cospinors too. 
On account of Example \ref{exaTrivDirac}, one already has the inclusion $\Ps(D\Mb)\subset\sVan$. 
For the converse inclusion, take $\tau\in\sVan$ and notice that $\ips{\tau}{\Es\sigma}=S_\tau(\Es\sigma)=0$ 
for all $\sigma\in\sectc(D\Mb)$ since $\Es$ is the advanced-minus-retarded operator for $\Ps$. 
$\Ps$ is formally self-adjoint with respect to $\ips{\cdot}{\cdot}$ as shown in Proposition \ref{prpGreenHypDirac}, 
therefore, recalling the properties of retarded and advanced Green operators, 
we have $\ips{\Es\tau}{\sigma}=-\ips{\tau}{\Es\sigma}=0$ for all $\sigma\in\sectc(D\Mb)$, 
hence $\Es\tau=0$ due to the non-degeneracy of $\ips{\cdot}{\cdot}$. 
Recalling \eqref{eqSCExactSeq}, one finds $\upsilon\in\sectc(D\Mb)$ such that $\Ps\upsilon=\tau$, 
thus showing that $\sVan\subset\Ps(\sectc(D\Mb))$. 
\end{rem}

So far, we determined the spaces $\sClObs$ and $\cClObs$ of classical observables for Dirac spinors and cospinors. 
Yet, to formulate the corresponding quantum field theory, one still needs suitable Hermitian structures 
in order to write down the usual anti-commutation relations for Dirac fields. 
This is the purpose of the next proposition. 

\begin{proposition}\label{prpHermFormDirac}
Consider a four-dimensional globally hyperbolic spacetime $\Mb=(\Mc,g,\ogth,\tgth)$ 
and take a co-frame $e=(e^\mu)_{\mu=0,\dots,3}$ on it. 
Let $\Ps$ and $\Pc$ denote the differential operators ruling the dynamics of spinors and respectively of cospinors. 
Introduce the corresponding advanced-minus-retarded operators $\Es$ and $\Ec$. 
Those defined below are non-degenerate Hermitian forms on $\sClObs$ and respectively on $\cClObs$: 
\begin{equation}\label{eqHermFormDirac}
\begin{aligned}
\sHerm:\sClObs\times\sClObs & \to\CC, & ([\sigma],[\tau]) & \mapsto-i\ips{\sigma}{\Es\tau},\\
\cHerm:\cClObs\times\sClObs & \to\CC, & ([\omega],[\zeta]) & \mapsto i\ipc{\omega}{\Ec\zeta},
\end{aligned}
\end{equation}
where the representatives $\sigma\in[\sigma]$, $\tau\in[\tau]$, 
$\zeta\in[\zeta]$ and $\omega\in[\omega]$ are chosen arbitrarily. 
Furthermore, following \eqref{eqAdjSDirac}, the antilinear isomorphism 
\begin{align}
A:\sClObs\to\cClObs, && [\tau]\mapsto[A\tau]
\end{align}
relates $\sHerm$ to $\cHerm$, namely one has $\cHerm(A[\sigma],A[\tau])=\sHerm([\tau],[\sigma])$. 
\end{proposition}

\begin{proof}
We shall discuss explicitly the spinor case. The argument in the case of cospinors is very similar. 
First of all, let us show that $\sHerm$ is a well-defined non-degenerate Hermitian form. 
As a starting point, consider the map 
\begin{equation}\label{eqAntiHermDirac}
\ips{\cdot}{\Es\cdot}:\sectc(D\Mb)\times\sectc(D\Mb)\to\CC.
\end{equation}
Since $\Es$ is linear, this map is sesquilinear as $\ips{\cdot}{\cdot}$ is. 
Furthermore, $\Ps$ is formally self-adjoint with respect to $\ips{\cdot}{\cdot}$ 
and $\Ps\circ\Es=0=\Es\circ\Ps$ on $\sectc(D\Mb)$. 
This entails that $\ips{\cdot}{\Es\cdot}$ vanishes 
whenever one of its arguments is of the form $\Ps\tau$ for any $\tau\in\sectc(D\Mb)$. 
%In fact, for all $\omega\in\sectc(D^\ast\Mb)$, one has 
%\begin{equation}
%\begin{aligned}
%\la\omega,\Es A^{-1}\Pc\zeta\ra & =\la\omega,\Es\Ps A^{-1}\zeta\ra=0,\\
%\la\Pc\zeta,\Es A^{-1}\omega\ra & =-\la\Ec\Pc\zeta,A^{-1}\omega\ra=0.
%\end{aligned}
%\end{equation}
It follows that the form defined in \eqref{eqAntiHermDirac} descends 
to the quotient $\sectc(D^\ast\Mb)/\Ps(\sectc(D^\ast\Mb))$. 
On account of \eqref{eqClObsDirac} and of \eqref{eqVanImPDirac}, 
the space of classical observables $\sClObs$ for spinors is exactly of this form, 
hence $\sHerm$ is a well-defined sesquilinear form on $\sClObs$, 
namely it is anti-linear in the first argument and linear in the second. 

The second part of the proof is devoted to showing that $\sHerm$ is Hermitian. 
This follows from the fact that $\ips{\cdot}{\cdot}$ provides a Hermitian form on $\sectc(D\Mb)$, 
{\em cf.} \eqref{eqHermInnProdDirac}. Specifically, given $\sigma,\tau\in\sectc(D\Mb)$, the following holds: 
\begin{equation}\label{eqHermCheck}
\ol{\sHerm([\sigma],[\tau])}=\ol{-i\ips{\sigma}{\Es\tau}}
=i\ips{\Es\tau}{\sigma}=-i\ips{\tau}{\Es\sigma}=\sHerm([\tau],[\sigma]).
\end{equation}
Notice that in the third equality we exploited the formal self-adjointness of $\Ps$ 
with respect to $\ips{\cdot}{\cdot}$, which entails that 
$\ips{\Es\sigma}{\tau}=-\ips{\sigma}{\Es\tau}$ for all $\sigma,\tau\in\sectc(D\Mb)$. 

We still have to prove that $\sHerm$ is non-degenerate. To this aim, 
consider $[\tau]\in\sClObs$ such that $\sHerm([\sigma],[\tau])=0$ for all $[\sigma]\in\sClObs$. 
Our goal is to show that this condition implies $[\tau]=0$. 
In fact, one deduces that $\ips{\sigma}{\Es\tau}$ has to vanish for all $\sigma\in\sectc(D\Mb)$. 
Therefore, by non-degeneracy of the pairing $\ips{\cdot}{\cdot}$ 
between $\sectc(D\Mb)$ and $\sect(D\Mb)$, one deduces that $\Es\tau=0$, 
hence there exists $\upsilon\in\sectc(D\Mb)$ such that $\Ps\upsilon=\tau$. This proves that $[\tau]=0$. 

The last part of the proof focuses on the relation between $\sHerm$ and $\cHerm$. 
Since $A\circ\Ps=\Pc\circ A$, it follows that $A:\sClObs\to\cClObs$ is well-defined. 
Furthermore, this is an anti-linear isomorphism of vector spaces 
since $A:D\Mb\to D^\ast\Mb$ is an anti-linear vector bundle isomorphism. 
For each $\sigma,\tau\in\sectc(D\Mb)$, one has the following chain of equalities: 
\begin{equation*}
\begin{aligned}
\cHerm(A[\sigma],A[\tau]) & =i\ipc{A\sigma}{\Ec A\tau}=i\ipc{A\sigma}{A\Es\tau}\\
& =i\ips{\Es\tau}{\sigma}=-i\ips{\tau}{\Es\sigma}=\sHerm([\tau],[\sigma]).
\end{aligned}
\end{equation*}
For the second equality, we used the identity $\Ec\circ A=A\circ\Es$ on $\sectc(D\Mb)$, 
which follows from $A\circ\Ps=\Pc\circ A$ on $\sect(D\Mb)$.
\end{proof}

\begin{rem}
The general theory of Green hyperbolic operators provides an isomorphism 
between the space $\sClObs$ of classical observables for spinors 
and the space $\sSolsc$ of on-shell spinors with spacelike compact support, see Proposition \ref{prpSCSol}. 
This isomorphism is realized by the advanced-minus-retarded operator $\Es$ for $\Ps$, 
which is the differential operator ruling the dynamics for Dirac spinors. 
Similarly, $\Ec$, the advanced-minus-retarded operator corresponding to $\Pc$, 
provides an isomorphism between $\cClObs$ and $\cSolsc$: 
\begin{equation*}
\begin{aligned}
I_\mathrm{s}:\sClObs & \to\sSolsc, & [\tau]\to\Es\tau,\\
I_\mathrm{c}:\cClObs & \to\cSolsc, & [\zeta]\to\Ec\zeta.
\end{aligned}
\end{equation*}
$I^\mathrm{s}$ and $I^\mathrm{c}$ become isomorphisms of Hermitian spaces 
as soon as $\sSolsc$ and $\cSolsc$ are endowed with the usual Hermitian structures for Dirac fields 
written in terms of the initial data on a spacelike Cauchy surface $\Sigma$ 
for the four-dimensional globally hyperbolic spacetime $\Mb=(\Mc,g,\ogth,\tgth)$. 
Denoting with $\nb$ the future-pointing unit normal vector field on $\Sigma$ 
and with $d\Sigma$ the volume form naturally induced on $\Sigma$, 
one introduces the following non-degenerate Hermitian forms on $\sSolsc$ and on $\cSolsc$: 
\begin{align}
\sHermSol:\sSolsc\times\sSolsc & \to\CC, & 
\sHermSol(\sigma,\tau) & =\int_\Sigma(A\sigma)(\slashed{\nb}\tau)\,d\Sigma,\\
\cHermSol:\cSolsc\times\cSolsc & \to\CC, & 
\cHermSol(\omega,\zeta) & =\int_\Sigma\zeta(\slashed{\nb}A^{-1}\omega)\,d\Sigma,
\end{align}
where $\slashed{\nb}=\gamma(\nb)$ denotes the section over $\Sigma$ of $\Endo(D\Mb)$ 
obtained evaluating the $\Endo(D\Mb)$-valued one-form $\gamma$ 
on the vector field $\nb$ at each point of $\Sigma$. 
One can prove that $I_\mathrm{s}$ preserves the Hermitian structures by mimicking the strategy used in \eqref{eqSymplSolProofScalar} for the real scalar field and by relying on the calculation presented in \eqref{eqNablaSlashAntiAdjProof}. 
More explicitly, given $\sigma,\tau\in\sectc(D\Mb)$, one finds the following: 
\begin{equation*}
\begin{aligned}
\sHerm([\sigma],[\tau]) & =-i\int\limits_{J_\Mb^-(\Sigma)}\big(A(\Ps\Es^+\sigma)\big)(\Es\tau)\,\dvol_\Mb
-i\int\limits_{J_\Mb^+(\Sigma)}\big(A(\Ps\Es^-\sigma)\big)(\Es\tau)\,\dvol_\Mb\\
& =-\int\limits_{J_\Mb^-(\Sigma)}\dd\ast\Big(\big(A(\Es^+\sigma)\big)\big(\gamma(\Es\tau)\big)\Big)
-\int\limits_{J_\Mb^+(\Sigma)}\dd\ast\Big(\big(A(\Es^-\sigma)\big)\big(\gamma(\Es\tau)\big)\Big)\\
& =-\int_\Sigma\big(A(\Es^+\sigma)\big)\big(\slashed{\nb}(\Es\tau)\big)\,d\Sigma
+\int_\Sigma\big(A(\Es^-\sigma)\big)\big(\slashed{\nb}(\Es\tau)\big)\,d\Sigma\\
& =\sHermSol(\Es\sigma,\Es\tau).
\end{aligned}
\end{equation*}
Notice that, after the integration by parts of the terms involving $\nss$ has been performed, 
only boundary terms are left due to the fact that $\Ps\Es\tau=0$. 
The case of cospinors follows suit. 
\end{rem}

At this stage we have the Hermitian forms $\sHerm$ and $\cHerm$ defined on the spaces 
of classical observables $\sClObs$ and $\cClObs$ respectively for spinors and for cospinors. 
Therefore, we can consider the Hermitian spaces $(\sClObs,\sHerm)$ and $(\cClObs,\cHerm)$. 
Furthermore, according to Proposition \ref{prpHermFormDirac}, $A:\sClObs\to\cClObs$ establishes 
a strict relation between the two Hermitian structures $\sHerm$ and $\cHerm$. 
These Hermitian spaces and their relation are exactly the data needed in order to pass from the classical Dirac field 
to the corresponding quantum counterpart. However, before turning our attention to the quantum case, 
we would like to investigate some properties of the Hermitian spaces $(\sClObs,\sHerm)$ and $(\cClObs,\cHerm)$. 

\begin{theorem}\label{thmClPropDirac}
Let $\Mb=(\Mc,g,\ogth,\tgth)$ be a four-dimensional globally hyperbolic spacetime  
and take a co-frame $e=(e^\mu)_{\mu=0,\dots,3}$ on it. 
Let $(\sClObs,\sHerm)$ and $(\cClObs,\cHerm)$ be the Hermitian spaces of classical observables defined above 
respectively for spinors and for cospinors. The following properties hold: 
\begin{description}
\item[{\bf Causality}] The Hermitian structures vanish on pairs of observables localized in causally disjoint regions. 
More precisely, let $\sigma,\tau\in\sectc(D\Mb)$ be such that $\supp\sigma\cap J(\supp\tau)=\emptyset$. 
Then $\sHerm([\sigma],[\tau])=0$. Similarly, taking $\omega,\zeta\in\sectc(D^\ast\Mb)$ 
such that $\supp\omega\cap J_\Mb(\supp\zeta)=\emptyset$, one has $\cHerm([\omega],[\zeta])=0$. 
\item[{\bf Time-slice axiom}] Let $\Oc\subset\Mc$ be a globally hyperbolic open neighborhood 
of a spacelike Cauchy surface $\Sigma$ for $\Mb$, namely $\Oc$ is an open neighborhood of $\Sigma$ in $\Mc$ 
containing all causal curves for $\Mb$ whose endpoints lie in $\Oc$. 
In particular, the restriction of $\Mb$ to $\Oc$ provides 
a globally hyperbolic spacetime $\Ob=(\Oc,g\vert_\Oc,\ogth\vert_\Oc,\tgth\vert_\Oc)$. 
Furthermore, as a co-frame on $\Ob$, we consider the restriction of the co-frame $e$ on $\Mb$. 
Denote with $(\sClObs_\Mb,\sHerm_\Mb)$ and with $(\sClObs_\Ob,\sHerm_\Ob)$ the Hermitian spaces 
of observables for spinors respectively over $\Mb$ and over $\Ob$. Similarly, let $(\cClObs_\Mb,\cHerm_\Mb)$ 
and $(\cClObs_\Ob,\cHerm_\Ob)$ denote the Hermitian spaces of observables for cospinors 
respectively over $\Mb$ and over $\Ob$. Then the maps $\Ls:\sClObs_\Ob\to\sClObs_\Mb$ 
and $\Lcc:\cClObs_\Ob\to\cClObs_\Mb$, defined by $\Ls[\tau]=[\tau]$ for all $\tau\in\sectc(D\Ob)$ 
and by $\Lcc[\zeta]=[\zeta]$ for all $\zeta\in\sectc(D^\ast\Ob)$, are isomorphisms of Hermitian spaces.\footnote{
The sections on the right-hand-side in the definitions of $\Ls$ and of $\Lcc$ 
are the extensions by zero to the whole spacetime of the sections which appear on the left-hand-side.} 
\end{description}
\end{theorem}

\begin{proof}
The proof of this theorem follows slavishly that of Theorem \ref{thmClPropScalar} for the real scalar case. 
In fact, the only difference is that we are replacing symplectic structures with Hermitian ones. 
Apart form that, the proof presented there holds in this case as well. 
In fact, the argument relies only on the Green hyperbolicity of the differential operators ruling the dynamics 
and indeed both $\Ps$ and $\Pc$ have this property according to Proposition \ref{prpGreenHypDirac}. 
%We provide proofs only for the case of spinors, the case of cospinors being basically the same. 
%Let us start form causality. We consider $\sigma,\tau\in\sectc(D\Mb)$ such that 
%$\supp\sigma\cap J(\supp\tau)=\emptyset$. Recalling the support properties 
%of the advanced-minus-retarded operator $\Es$ for $\Ps$ 
%and the definition of $\sHerm$ given in \eqref{eqHermFormDirac}, 
%we conclude that $\sHerm([\sigma],[\tau])=0$ as expected. 
%
%For the time slice axiom, we first show that $\Ls:\sClObs_\Ob\to\cClObs_\Mb$ is a well-defined linear map 
%preserving the Hermitian structures $\sHerm_\Ob$ and $\cHerm$. 
%Since $\sVan_\Ob=\Ps(\sectc(D\Ob))$ and the extension by zero of a section in $\Ps(\sectc(D\Ob))$ 
%is indeed a section in $\Ps(\sectc(D\Mb))$, the map $\Ls$ is well-defined. 
%From its definition, one immediately deduces that the map is linear and that it preserves the Hermitian structure. 
%In fact, given two sections $\sigma,\tau\in\sectc(D\Ob)$ 
%and denoting by the same symbols also their extensions by zero to the full spacetime, one finds 
%\begin{equation}
%\begin{aligned}
%\sHerm_\Mb(\Ls[\sigma],\Ls[\tau]) & =-i\int_\Mc(A\sigma)(\Es\tau)\,\dvol_\Mb\\
%& =-i\int_\Oc(A\sigma)(\Es\tau)\,\dvol_\Ob=\sHerm_\Ob([\sigma],[\tau]).
%\end{aligned}
%\end{equation}
%Notice that the restriction in the domain of integration is due to fact that, by construction, 
%the section $\sigma$ appearing in the second step vanishes outside of $\Oc$. 
%To prove that $\Ls$ is actually an isomorphism, we look for an inverse $\Ks:\
\end{proof}

\subsubsection{Quantum field theory}
Given a four-dimensional globally hyperbolic spacetime $\Mb=(\Mc,g,\ogth,\tgth)$ 
and choosing a co-frame $e=(e^\mu)_{\mu=0,\dots,3}$ on it, 
in the previous section we were able to construct all the kinematical and dynamical objects 
related to the classical theory of the Dirac field. 
In particular, we obtained two Hermitian spaces of classical observables for the Dirac field on $\Mb$. 
The first one, $(\sClObs,\sHerm)$, is used to test on-shell spinors, 
while the second one, $(\cClObs,\cHerm)$, pertains to cospinors. 
Furthermore, the two Hermitian spaces are related by an anti-linear isomorphism $A:\sClObs\to\cClObs$, 
which satisfies $\cHerm(A[\sigma],A[\tau])=\sHerm([\tau],[\sigma])$ for all $[\sigma],[\tau]\in\sClObs$. 
Let us stress that these spaces faithfully represent a class of linear functionals defined on on-shell Dirac fields, 
which is rich enough to distinguish between different field configurations. 
These properties motivate our interpretation of $\sClObs$ and $\cClObs$ 
as spaces of classical observables for the Dirac field. 

Now we want to switch from the classical field theoretical description to its quantum counterpart. As for the scalar case, we shall only construct a suitable algebra of observables, omitting any discussion concerning algebraic states, a topic which is addressed in \cite{IV}. This result is achieved by considering the unital $\ast$-algebra $\Ac$ defined as follows. 
Starting from the unital $\ast$-algebra freely generated over $\CC$ 
by the symbols $\II$, $\Phi([\tau])$ and $\Psi([\zeta])$ for all $[\tau]\in\sClObs$ and for all $[\zeta]\in\cClObs$, 
we impose the relations listed below, thus obtaining the sought unital $\ast$-algebra $\Ac$: 
\begin{align}
\Phi(a[\sigma]+b[\tau]) & =a\Phi([\sigma])+b\Phi([\tau]),\label{eqLinearityDirac}\\
\Phi([\sigma])^\ast & =\Psi(A[\sigma]),\label{eqInvolutionDirac}\\
\Phi([\sigma])\cdot\Phi([\tau])+\Phi([\tau])\cdot\Phi([\sigma]) & =0,\label{eqCAR1Dirac}\\
\Psi([\omega])\cdot\Psi([\zeta])+\Psi([\zeta])\cdot\Psi([\omega]) & =0,\label{eqCAR2Dirac}\\
\Psi([\zeta])\cdot\Phi([\tau])+\Phi([\tau])\cdot\Psi([\zeta]) & =\cHerm(A[\tau],[\zeta])\II.\label{eqCAR3Dirac}
\end{align}
These relations must hold for all $a,b\in\CC$, for all $[\sigma],[\tau]\in\sClObs$ 
and for all $[\omega],[\zeta]\in\cClObs$. 
In view of \eqref{eqLinearityDirac} the map $\Phi:\sClObs\to\Ac$, $[\tau]\mapsto\Phi([\tau])$ is linear. 
On account of \eqref{eqInvolutionDirac} and $A:\sClObs\to\cClObs$ being anti-linear, 
the map $\Psi:\cClObs\to\Ac$, $[\zeta]\mapsto\Psi([\zeta])$, is linear too. 
To conclude, \eqref{eqCAR1Dirac}, \eqref{eqCAR2Dirac} and \eqref{eqCAR3Dirac} 
provide the canonical anti-commutation relations (CAR) for the Dirac field. 
A more concrete construction can be obtained mimicking the one for the real scalar field, 
see Section \ref{subsubQuantumScalar} and \cite{Araki70}. 
Specifically, we consider the vector space 
$A=\bigoplus_{k\in\NN_0}(\sClObs\oplus\cClObs)^{\otimes k}$.\footnote{As usual, 
the component of the direct sum corresponding to the degree $k=0$ is simply $\CC$.} 
This is endowed with the product specified $\cdot:A\times A\to A$ defined by 
\begin{align}
\{u_k\}\cdot\{v_k\}=\{w_k\}, && w_k=\sum_{i+j=k}u_i\otimes v_j. 
\end{align}
Clearly, endowing $A$ with $\cdot$ provides a unital algebra, whose unit is given by $\II=\{1,0,\dots\}$. 
The generators of this algebra are 
\begin{align}
\Phi([\tau])=\left\{0,\left(\begin{matrix} [\tau] \\ 0 \end{matrix}\right),0,\dots\right\}, 
&& \Psi([\zeta])=\left\{0,\left(\begin{matrix} 0 \\ [\zeta] \end{matrix}\right),0,\dots\right\}, 
\end{align}
for all $[\tau]\in\sClObs$ and for all $[\zeta]\in\cClObs$. 
So far, the construction is almost identical to the one for the real scalar field. 
The only difference is that we replaced the complexification of the space of classical observables for the scalar field 
with the direct sum of the spaces of classical observables for spinors and for cospinors. 
The involution $\ast:B\to B$ is implemented by means of the anti-linear isomorphism $A:\sClObs\to\cClObs$: 
\begin{equation*}
\begin{aligned}
& \left\{0,\dots,0,\left(\begin{matrix} [\tau_1] \\ [\zeta_1] \end{matrix}\right)\otimes
\cdots\otimes\left(\begin{matrix} [\tau_k] \\ [\zeta_k] \end{matrix}\right),0,\dots\right\}^\ast\\
& \qquad=\left\{0,\dots,0,\left(\begin{matrix} A^{-1}[\zeta_k] \\ A[\tau_k] \end{matrix}\right)\otimes
\cdots\otimes\left(\begin{matrix} A^{-1}[\zeta_1] \\ A[\tau_1] \end{matrix}\right),0,\dots\right\},
\end{aligned}
\end{equation*}
for all $k\in\NN_0$, for all $[\tau_1],\dots,[\tau_k]\in\sClObs$ and for all $[\zeta_1],\dots,[\zeta_k]\in\cClObs$. 
As always, $\ast$ is extended to all elements of $B$ by anti-linearity, thus turning $B$ into a unital $\ast$-algebra. 
The canonical anti-commutation relations are implemented taking the quotient of $B$ 
by the two-sided $\ast$-ideal $I$ of $B$ generated by the elements listed below: 
\begin{align}
\Phi([\sigma])\cdot\Phi([\tau]) & +\Phi([\tau])\cdot\Phi([\sigma]),\\
\Psi([\omega])\cdot\Psi([\zeta]) & +\Psi([\zeta])\cdot\Psi([\omega]),\\
\Psi([\zeta])\cdot\Phi([\tau]) & +\Phi([\tau])\cdot\Psi([\zeta])-\cHerm(A[\tau],[\zeta])\II,
\end{align}
for all $[\sigma],[\tau]\in\sClObs$ and for all $[\omega],[\zeta]\in\cClObs$. 
The unital $\ast$-algebra $\Ac=B/I$ resulting from the quotient is a concrete realization 
of the one presented in the first part of the present section. 

Having established the algebra $\Ac$ describing the quantum theory of the free Dirac field 
on the four-dimensional globally hyperbolic spacetime $\Mb$, 
we would like to investigate some of its properties, as well as its relation 
to the traditional presentation of the quantum Dirac field. 

\begin{rem}
Let us mention that, similarly to the scalar case, see Remark \ref{remQuantumScalar}, 
the Dirac quantum field theory presented above reduces to the one usually 
found in any undergraduate textbook on quantum field theory 
as soon as $\Mb$ is Minkowski spacetime. 
This can be seen by means of a suitable Fourier expansion of the solutions to the field equations. 
\end{rem}

The properties of the classical theory of the Dirac field, which were investigated in Theorem \ref{thmClPropDirac}, 
have counterparts at the quantum level. We conclude this section analyzing this aspect. 

\begin{theorem}
Let $\Mb=(\Mc,g,\ogth,\tgth)$ be a four-dimensional globally hyperbolic spacetime 
and take a co-frame $e=(e^\mu)_{\mu=0,\dots,3}$ on it. 
Let $\Ac$ be the unital $\ast$-algebra of quantum observables for the Dirac field on $\Mb$. 
The following properties hold: 
\begin{description}
\item[{\bf Causality}] The elements of $\Ac$ localized in causally disjoint regions anti-commute. 
To wit, let $\zeta\in\sectc(D^\ast\Mb)$ and $\tau\in\sectc(D\Mb)$ be such that 
$\supp\zeta\cap J_\Mb(\supp\tau)=\emptyset$. 
Then $\Psi([\zeta])\cdot\Phi([\tau])+\Phi([\tau])\cdot\Psi([\zeta])=0$. 
In particular, the even subalgebra $\Ac^\even$ of $\Ac$, 
whose elements are finite linear combinations of products of an even number of generators of $\Ac$, 
fulfills the bosonic version of causality, 
namely the elements of $\Ac^\even$ localized in causally disjoint regions commute. 
\item[{\bf Time-slice axiom}] Let $\Oc\subset\Mc$ be a globally hyperbolic open neighborhood 
of a spacelike Cauchy surface $\Sigma$ for $\Mb$, namely $\Oc$ is an open neighborhood of $\Sigma$ in $\Mc$ 
containing all causal curves for $\Mb$ whose endpoints lie in $\Oc$. 
In particular, the restriction of $\Mb$ to $\Oc$ provides 
a globally hyperbolic spacetime $\Ob=(\Oc,g\vert_\Oc,\ogth\vert_\Oc,\tgth\vert_\Oc)$. 
Furthermore, as a co-frame on $\Ob$, we consider the restriction of the co-frame $e$ on $\Mb$. 
Denote with $\Ac_\Mb$ and with $\Ac_\Ob$ the unital $\ast$-algebras of observables for the Dirac field 
respectively over $\Mb$ and over $\Ob$. Then the map $I:\Ac_\Ob\to\Ac_\Mb$, 
defined on generators by $I(\Phi([\tau]))=\Phi(\Ls[\tau])$ for all $[\tau]\in\sClObs_\Ob$ 
and by $I(\Psi([\zeta]))=\Psi(\Lcc[\zeta])$ for all $[\zeta]\in\cClObs_\Ob$, 
is an isomorphism of unital $\ast$-algebras. 
Recall that the Hermitian isomorphisms $\Ls:\sClObs_\Ob\to\sClObs_\Mb$ 
and $\Lcc:\cClObs_\Ob\to\cClObs_\Mb$ were introduced in Theorem \ref{thmClPropDirac}. 
\end{description}
\end{theorem}

\begin{proof}
Given $\zeta\in\sectc(D^\ast\Mb)$ and $\tau\in\sectc(D\Mb)$ 
such that $\supp\zeta\cap J(\supp\tau)=\emptyset$, from Theorem \ref{thmClPropDirac}, 
one deduces that $\cHerm(A[\tau],[\zeta])=0$. Therefore, recalling \eqref{eqCAR3Dirac}, 
one concludes that $\Psi([\zeta])\cdot\Phi([\tau])+\Phi([\tau])\cdot\Psi([\zeta])=0$. 
Let us now consider three generators $G_1,G_2,G_3$ of $\Ac$ (they can be 
either of the form $\Phi([\tau])$ for $[\tau]\in\sClObs$ or of the from $\Psi([\zeta])$ for $[\zeta]\in\cClObs$). 
We assume that $G_1$ and $G_2$ are localized in a region 
which is causally disjoint from the one where $G_3$ is localized. 
On account of the first part of this theorem, 
we deduce that $G_i\cdot G_3+G_3\cdot G_i=0$ for $i=1,2$. 
The following chain of equalities follows from the last identity: 
\begin{equation*}
\begin{aligned}
(G_1\cdot G_2)\cdot G_3-G_3\cdot(G_1\cdot G_2) & =G_1\cdot G_2\cdot G_3+G_1\cdot G_3\cdot G_2\\
& \quad-G_1\cdot G_3\cdot G_2-G_3\cdot G_1\cdot G_2=0.
\end{aligned}
\end{equation*}
This already entails that all elements of $\Ac^\even$ commute with all elements of $\Ac$ 
provided that they are localized in causally disjoint regions. The claim follows as a special case. 

The quantum time-slice axiom follows directly from its classical counterpart. 
The procedure is very similar to the scalar case, see Theorem \ref{thmQuantumScalar}. 
In fact, $I:\Ac_\Ob\to\Ac_\Mb$ is a homomorphism of unital $\ast$-algebras by definition 
and, moreover, we can introduce an inverse $I^{-1}:\Ac_\Mb\to\Ac_\Ob$ of $I$ 
simply setting $I^{-1}\Phi([\tau])=\Phi(\Ls^{-1}[\tau])$ for all $[\tau]\in\sClObs$ 
and $I^{-1}\Psi([\zeta])=\Psi(\Lcc^{-1}[\zeta])$ for all $[\zeta]\in\cClObs$. 
It is straightforward to check that $I^{-1}\circ I=\id_{\Ac_\Ob}$ and $I\circ I^{-1}=\id_{\Ac_\Mb}$, 
so that $I^{-1}$ is actually the inverse of $I$ and then $I$ is an isomorphism of unital $\ast$-algebras as claimed. 
\end{proof}

\subsection{The Proca field}
The last example we shall analyze is the Proca field over globally hyperbolic spacetimes. 
We shall adopt the same approach used in the previous cases. 
Specifically, we shall start investigating the properties of the differential operator 
which rules the dynamics of the Proca field. After that, we shall introduce a suitable space of classical observables. 
In particular, we want a space of sections which can be used to define linear functionals on on-shell Proca fields. 
As usual, we shall require that the functionals obtained are sufficiently many 
to distinguish between different on-shell configurations. Furthermore, we want to get rid of the redundancies 
which might be contained in the space of sections we use to produce functionals. 
As soon as these requirements are achieved, we shall interpret the result 
as a space of classical observables for the Proca field. In fact, this space faithfully represents 
a class of functionals defined on-shell, which is rich enough to detect any field configuration. 
Then we shall endow this space with a symplectic structure, 
which will play a central role in the prescription to quantize the classical Proca field. 
This topic will be addressed in the last part of this section. 
Before starting our analysis, let us mention some references where the Proca field has been studied 
using the language of algebraic quantum field theory. 
These are \cite{Furlani:1999kq, Fewster:2003ey, Dappiaggi:2011ms}. 

\subsubsection{Dynamics and classical observables}
Let us consider an $n$-dimensional globally hyperbolic spacetime $\Mb=(\Mc,g,\ogth,\tgth)$. 
Unless stated otherwise, from now on, $\Mb$ shall be kept fixed. The off-shell configurations for the Proca field 
are sections of the cotangent bundle $T^\ast\Mc$, namely one-forms over $\Mc$. 
Adopting a standard convention, we shall denote the space of $k$-forms by $\f^k(\Mc)$. 
Let us remind the reader that differential forms form a graded algebra 
with respect to the standard wedge product $\wedge:\f^k(\Mc)\times\f^{k^\prime}(\Mc)\to\f^{k+k^\prime}(\Mc)$. 
To introduce the dynamics, we need two operations on forms, namely the differential 
$\dd:\f^k(\Mc)\to\f^{k+1}(\Mc)$ and the Hodge dual $\ast:\f^k(\Mc)\to\f^{n-k}(\Mc)$. 
$\dd$ is defined simply out of the differentiable structure on $\Mc$, see \cite[Section 1.1]{BT82}, 
while $\ast$ depends also on the metric $g$ and the orientation $\ogth$, see \cite[Section 3.3]{Jos11}. 
For our purposes it is enough to mention that $\dd$ is a graded derivative 
with respect to the wedge product $\wedge$, 
that $\dd\dd=0$ and that $\ast$ is an isomorphism, hence $\ast^{-1}$ is well-defined. 
For further details on the theory of differential forms, see \cite{BT82}. 
$\dd$ and $\ast$ enable us to introduce the codifferential: 
\begin{align}
\de:\f^k(\Mc)\to\f^{k-1}(\Mc), && \de=(-1)^k\ast^{-1}\dd\ast.
\end{align}
Notice that $\de\de=0$ due to $\dd\dd=0$. 
Introducing the symmetric pairing $\ipf{\cdot}{\cdot}$ between $k$-forms, defined by 
\begin{equation*}
\ipf{\alpha}{\beta}=\int_\Mc\alpha\wedge\ast\beta,
\end{equation*}
where $\alpha,\beta\in\f^k(\Mc)$ have supports with compact intersection, 
one can prove that $\de$ is the formal adjoint of $\dd$ with respect to $\ipf{\cdot}{\cdot}$, 
meaning that $\ipf{\alpha}{\de\beta}=\ipf{\dd\alpha}{\beta}$ for all $\alpha\in\f^k(\Mc)$ 
and $\beta\in\f^{k+1}(\Mc)$ such that $\supp\alpha\cap\supp\beta$ is compact. 
In fact, applying Stokes' theorem, one finds 
\begin{equation*}
\begin{aligned}
\ipf{\dd\alpha}{\beta}-\ipf{\alpha}{\de\beta} & =\int_\Mc(\dd\alpha\wedge\ast\beta-\alpha\wedge\ast\de\beta)\\
& =\int_\Mc(\dd\alpha\wedge\ast\beta+(-1)^k\alpha\wedge\dd\ast\beta)\\
& =\int_\Mc\dd(\alpha\wedge\ast\beta)=0. 
\end{aligned}
\end{equation*}

After these preliminaries, we are ready to introduce the Proca equation over $\Mb$ for a form $A\in\f^1(\Mc)$: 
\begin{equation*}
-\de\dd A+m^2A=0
\end{equation*}
where $m^2\in\RR\setminus\{0\}$. As for the case of a real scalar field, all our results are valid for all possible values of the mass. Using the abstract index notation, the Proca equation reads 
\begin{equation*}
\nabla^a(\nabla_aA_b-\nabla_bA_a)+m^2A_b=0.
\end{equation*}
We introduce the second-order linear differential operator 
$P=-\de\dd+m^2\id_{\f^1(\Mc)}$. With this definition, the Proca equation can be rewritten as $PA=0$. 
Since $\de$ is the formal adjoint of $\dd$ with respect to $\ipf{\cdot}{\cdot}$, 
it follows that $P:\f^1(\Mc)\to\f^1(\Mc)$ is formally self-adjoint: 
\begin{equation}\label{eqSelfAdjProca}
\ipf{P\alpha}{\beta}=-\ipf{\dd\alpha}{\dd\beta}+m^2\ipf{\alpha}{\beta}=\ipf{\alpha}{P\beta}, 
\end{equation}
for all $\alpha,\beta\in\f^k(\Mc)$ whose supports have compact overlap. 
In the next proposition we show that $P$ is Green hyperbolic by exhibiting its retarded and advanced Green operators,
see Definition \ref{E+E-}. 

\begin{proposition}\label{prpGreenHypProca}
Let $\Mb=(\Mc,g,\ogth,\tgth)$ be an $n$-dimensional globally hyperbolic spacetime 
and let $Q=-m^{-2}\dd\de+\id_{\f^1(\Mc)}:\f^1(\Mc)\to\f^1(\Mc)$. 
The second order linear differential operator $P:\f^1(\Mc)\to\f^1(\Mc)$, 
which rules the dynamics of the Proca field on $\Mb$, is formally self-adjoint 
with respect to $\ipf{\cdot}{\cdot}$. Furthermore it is Green hyperbolic. 
In particular, its retarded and advanced Green operators are given by $E^\pm=QF^\pm$, 
where $F^+$ and $F^-$ denote the retarded and advanced Green operators 
for the normally hyperbolic operator $R=PQ:\f^1(\Mc)\to\f^1(\Mc)$. 
\end{proposition}

\begin{proof}
In \eqref{eqSelfAdjProca} we have shown that $P$ is formally self-adjoint. 
Let us consider $R=PQ$. Recalling that $\dd\dd=0$, one finds that $R=-\de\dd-\dd\de+m^2\id_{\f^1(\Mc)}$. 
Therefore $R$ coincides with the Hodge-d'Alembert operator $-\de\dd-\dd\de$, 
up to a term of order zero in the derivatives. In particular, in local coordinates, 
the principal part of $R$ is of the form $g^{\mu\nu}\partial_\mu\partial_\nu$, 
hence $R$ is normally hyperbolic. On account of \cite{Bar} and \cite[Chapter 3]{BGP}, 
$R$ admits unique retarded and advanced Green operators $F^+$ and $F^-$. 
To conclude the proof we have to show that $E^+=QF^+$ and $E^-=QF^-$ 
are retarded and advanced Green operators for $P$. 
Since $Q$ is a linear differential operator, it cannot enlarge the support. 
Therefore, $E^\pm$ inherits the correct support property for a retarded/advanced Green operator from $F^\pm$. 
Furthermore, for all $\alpha \in \f^1_{pc/fc}(\Mc)$, one has $PE^\pm\alpha=RF^\pm\alpha=\alpha$. 
It remains only to check that $E^\pm P\alpha=\alpha$ for all $\alpha \in \f^1_{pc/fc}(\Mc)$. 
Exploiting the formal self-adjointness of $P$ and keeping in mind the first part of the proof, 
one gets the following chain of identities for all $\beta\in\fc^1(\Mc)$: 
\begin{equation*}
\begin{aligned}
\ipf{\beta}{E^\pm P\alpha} & =\ipf{PE^\mp\beta}{E^\pm P\alpha}=\ipf{E^\mp\beta}{PE^\pm P\alpha}\\
& =\ipf{E^\mp\beta}{P\alpha}=\ipf{PE^\mp\beta}{\alpha}=\ipf{\beta}{\alpha}.
\end{aligned}
\end{equation*}
Since $\ipf{\cdot}{\cdot}$ provides a non-degenerate bilinear pairing between $\fc^1(\Mc)$ and $\f^1(\Mc)$, 
we conclude that $E^\pm P\alpha=\alpha$ for all $\alpha \in \f^1_{pc/fc}(\Mc)$, 
hence $E^+$ and $E^-$ are retarded and advanced Green operators for $P$, 
which is consequently Green hyperbolic. 
\end{proof}

Due to Proposition \ref{prpGreenHypProca}, we can apply the general theory of Green hyperbolic operators 
presented in Section \ref{sec:2} to the operator $P$ ruling the dynamics of the Proca field. 
In particular we find that $\ipf{E^\pm\alpha}{\beta}=\ipf{\alpha}{E^\mp\beta}$ 
for all $\alpha, \beta \in \fc^1(\Mc)$ and we can introduce the advanced-minus-retarded operator $E=E^--E^+$, 
which enables us to represent any solution starting from a one-form with timelike compact support. 

We have completed the analysis of the dynamics of the Proca field 
over the $n$-dimensional globally hyperbolic spacetime $\Mb$. 
In the following we shall focus on the construction of a suitable space of classical observables. 
Exploiting the non-degenerate bilinear pairing $\ipf{\cdot}{\cdot}$ between $\fc^1(\Mc)$ and $\f^1(\Mc)$, 
we can introduce a family of linear functionals on off-shell configurations. 
In fact, given $\alpha\in\f^1(\Mc)$, we consider 
\begin{align}\label{eqFunctProca}
F_\alpha:\f^1(\Mc)\to\RR, && A\mapsto\ipf{\alpha}{A}.
\end{align}
The fact that $\ipf{\cdot}{\cdot}$ is non-degenerate has two consequences. 
The first one is that we can identify the vector space $\{F_\alpha:\alpha\in\fc^1(\Mc)\}$, 
formed by the functionals introduced in \eqref{eqFunctProca}, with $\fc^1(\Mc)$. 
The second one is that the mentioned space of functionals is sufficiently rich 
to distinguish between different off-shell configurations. 
In particular, on-shell configurations can be separated as well, 
therefore our first requirement for the space of classical observables is achieved by $\f^1(\Mc)$. 
Yet, as soon as we go on-shell, which corresponds to restricting the functionals defined above 
to field configurations $A\in\f^1(\Mc)$ satisfying the equation of motion $PA=0$, 
some of the functionals become trivial. 
Before presenting explicit examples of this kind of redundancy for certain elements of $\fc^1(\Mb)$, 
let us introduce: 
\begin{equation*}
\Sol=\{A\in\f^1(\Mc):\,PA=0\}.
\end{equation*}

\begin{example}\label{exaVanProca}
The situation here is basically the same as in Example \ref{exaVanScalar} for the scalar field. 
In fact, formal self-adjointness of $P$ with respect to $\ipf{\cdot}{\cdot}$ entails that 
$F_{P\alpha}(A)=F_\alpha(PA)$ for all $\alpha\in\fc^1(\Mc)$ and for all $A\in\f^1(\Mc)$. 
Therefore, we have $F_{P\alpha}(A)=0$ for $A\in\Sol$, 
thus showing that one-forms in $P(\fc^1(\Mc)))$ are redundant 
in the sense that they provide functionals which always vanish on-shell. 
\end{example}

As shown by the example above, $\fc^1(\Mc)$ does not provide a faithful way 
to represent the restrictions to on-shell configurations of the functionals defined in \eqref{eqFunctProca}. 
Therefore $\fc^1(\Mc)$ does not meet our second requirement to be identified 
with the space of classical observables for the Proca field. 
In order to circumvent this issue, we proceed as in the previous cases. 
Introducing the subspace 
\begin{equation*}
N=\{\alpha\in\fc^1(\Mc):\,F_\alpha(A)=0,\;\forall A\in\Sol\}\subset\fc^1(\Mc)
\end{equation*}
of those one-forms which produce functionals vanishing on-shell, we consider the quotient space 
\begin{equation*}
\ClObs=\fc^1(\Mc)/N.
\end{equation*}
Per construction $\ClObs$ has no redundancy left and therefore it represents faithfully the restrictions to $\Sol$ 
of the functionals in \eqref{eqFunctProca}. Notice that this representation is realized 
by sending each equivalence class $[\alpha]\in\ClObs$ to the functional $F_\alpha:\Sol\to\RR$ 
defined by any representative $\alpha\in[\alpha]$. 
This assignment is well-defined because two representatives of $[\alpha]$ 
differ by a one-form which produces a functional vanishing on-shell. 
Since the original space $\fc^1(\Mc)$ is sufficient to separate solutions, 
this is the case for $\ClObs$ too. These features motivate our interpretation of $\ClObs$ 
as the space of classical observables for the Proca field on the globally hyperbolic spacetime $\Mb$. 

To complete our analysis of the classical theory of the Proca field, we still have to endow $\ClObs$ 
with a symplectic structure, which will eventually enable us to quantize the model 
by means of canonical commutation relations. We shall prove first that $N=P(\fc^1(\Mc))$ 
and then Proposition \ref{prpSymplStructure} will provide the desired symplectic structure on $\ClObs$. 
The situation is again basically the same as in the scalar case. 
In fact, Example \ref{exaVanProca} provides the inclusion $P(\fc^1(\Mc))\subset N$ 
and we are left with the proof of the converse inclusion, which follows just from the Green hyperbolicity of $P$. 
Given $\alpha\in N$, $F_\alpha(E\beta)=0$ for all $\beta\in\fc^1(\Mc)$ due to $E\beta$ being a solution. 
Yet, this means that $\ipf{E\alpha}{\beta}=-\ipf{\alpha}{E\beta}=0$ for all $\beta\in\fc^1(\Mc)$, 
hence the non-degeneracy of $\ipf{\cdot}{\cdot}$ entails that $E\alpha=0$. 
Exploiting \eqref{eqSCExactSeq}, we find $\gamma\in\fc^1(\Mc)$ such that $P\gamma=\alpha$, 
thus proving the desired inclusion $N\subset P(\fc^1(\Mc))$. 
In particular, we have that $\ClObs$ is the same as the quotient $\fc^1(\Mc)/P(\fc^1(\Mc))$. 
Therefore, recalling Proposition \ref{prpSymplStructure}, we get a symplectic structure 
\begin{align}
\tau:\ClObs\times\ClObs\to\RR, && ([\alpha],[\beta])\mapsto\ipf{\alpha}{E\beta}.
\end{align}
In particular, we can regard $(\ClObs,\tau)$ as a symplectic space of classical observables for the Proca field 
over the globally hyperbolic spacetime $\Mb$. 

\begin{rem}
It is often customary to present the symplectic form as an integral over a spacelike Cauchy surface $\Sigma$ 
of the globally hyperbolic spacetime $\Mb=(\Mc,g,\ogth,\tgth)$. The integrand is given in terms 
of those data on $\Sigma$ which are needed to set up an initial value problem for the field equation of interest. 
Similarly to the scalar and Dirac cases, we show how to relate our approach to the latter one. 
Let us consider $\alpha,\beta\in\fc^1(\Mc)$ and note that $E\beta$ is a solution. 
We shall split the integral which defines $\tau([\alpha],[\beta])$ in two parts 
and we shall exploit the properties of the retarded and advanced Green operators to replace $\alpha$ 
with $PE^\pm\alpha$ in such a way that we are allowed to use Stokes' theorem: 
\begin{equation}\label{eqSymplSolProca}
\begin{aligned}
\tau([\alpha],[\beta]) 
& =\int_{J^-_\Mc(\Sigma)}(PE^+\alpha)\wedge\ast(E\beta)
+\int_{J^+_\Mc(\Sigma)}(PE^-\alpha)\wedge\ast(E\beta)\\
& =-\int_{J^-_\Mc(\Sigma)}\dd\big((E^+\alpha)\wedge\ast\dd(E\beta)-(E\beta)\wedge\ast\dd(E^+\alpha)\big)\\
& \quad-\int_{J^+_\Mc(\Sigma)}\dd\big((E^-\alpha)\wedge\ast\dd(E\beta)
-(E\beta)\wedge\ast\dd(E^-\alpha)\big)\\
& =\int_\Sigma\big((E\alpha)\wedge\ast\dd(E\beta)-(E\beta)\wedge\ast\dd(E\alpha)\big)\\
& =\int_\Sigma\Big(g\big(E\alpha,\imath_\nb\dd(E\beta)\big)
-g\big(E\beta,\imath_\nb\dd(E\alpha)\big)\Big)\,d\Sigma,
\end{aligned}
\end{equation}
where $d\Sigma$ is the naturally induced volume form on $\Sigma$, 
$\nb$ denotes the future-pointing unit normal vector field on $\Sigma$ 
and $\imath_\nb$ is the operator which inserts the vector field $\nb$ in the form to which it is applied. 
Notice that the integration by parts only gives boundary terms since $PE\beta=0$. 
Due to \eqref{eqSymplSolProca}, it is easy to realize that our symplectic space $(\ClObs,\tau)$ 
is isomorphic to the symplectic space $(\Solsc,\sigma)$ often considered in the literature, 
where $\Solsc$ denotes the space of on-shell configurations of the Proca field with spacelike compact support 
and $\sigma:\Solsc\times\Solsc\to\RR$ is the symplectic form defined for all $A,B\in\Solsc$ by 
\begin{equation*}
\sigma(A,B)=\int_\Sigma\big(g(A,\imath_\nb\dd B)-g(B,\imath_\nb\dd A)\big)\,d\Sigma.
\end{equation*}
\end{rem}

Before turning our attention to the quantum theory of the Proca field, 
we devote a few lines to examine some of the properties of the symplectic space $(\ClObs,\tau)$ 
of classical observables for the Proca field over the globally hyperbolic spacetime $\Mb$. 
Notice that we shall not provide the details of the proof since this would be 
nothing more than a slavish copy of the proof of Theorem \ref{thmClPropScalar}. 

\begin{theorem}\label{thmClPropProca}
Let $\Mb=(\Mc,g,\ogth,\tgth)$ be an $n$-dimensional globally hyperbolic spacetime 
and let $(\ClObs,\tau)$ be the symplectic space of classical observables introduced above for the Proca field. 
The following properties hold: 
\begin{description}
\item[{\bf Causality}] The symplectic structure vanishes on pairs of observables localized in causally disjoint regions. 
More precisely, let $\alpha,\beta\in\fc^1(\Mc)$ be such that $\supp \alpha\cap J_\Mb(\supp\beta)=\emptyset$. 
Then $\tau([\alpha],[\beta])=0$. 
\item[{\bf Time-slice axiom}] Let $\Oc\subset\Mc$ be a globally hyperbolic open neighborhood 
of a spacelike Cauchy surface $\Sigma$ for $\Mb$, namely $\Oc$ is an open neighborhood of $\Sigma$ in $\Mc$ 
containing all causal curves for $\Mb$ whose endpoints lie in $\Oc$. 
In particular, the restriction of $\Mb$ to $\Oc$ provides 
a globally hyperbolic spacetime $\Ob=(\Oc,g\vert_\Oc,\ogth\vert_\Oc,\tgth\vert_\Oc)$. 
Denote with $(\ClObs_\Mb,\tau_\Mb)$ and with $(\ClObs_\Ob,\tau_\Ob)$ the symplectic spaces of observables 
for the Proca field respectively over $\Mb$ and over $\Ob$. 
Then the map $L:\ClObs_\Ob\to\ClObs_\Mb$ defined by $L[\alpha]=[\alpha]$ for all $\alpha\in\fc^1(\Oc)$ 
is an isomorphism of symplectic spaces.\footnote{The differential form in the right-hand-side of the equation 
which defines $L$ is the extension by zero to the whole spacetime of the differential form 
which appears in the left-hand-side.} 
\end{description}
\end{theorem}

\subsubsection{Quantum field theory}
To complete our analysis of the Proca field, we present the quantization 
of the classical field theory developed in the previous section, 
which consists of a symplectic space $(\ClObs,\tau)$ of classical observables for the Proca field 
over a globally hyperbolic spacetime $\Mb=(\Mc,g,\ogth,\tgth)$. 
The quantization procedure is completely equivalent to the case of the real scalar field. 
For this reason we shall skip most of the details, referring the reader to Section \ref{subsubQuantumScalar}. 
We introduce the quantum theory of the Proca field in terms of the unital $\ast$-algebra $\Ac$ 
generated over $\CC$ by the symbols $\II$ and $\Phi([\alpha])$ for all classical observables $[\alpha]\in\ClObs$ 
and satisfying the relations listed below: 
\begin{align}
\Phi(a[\alpha]+b[\beta]) & =a\Phi([\alpha])+b\Phi([\beta]),\\
\Phi([\alpha])^\ast & =\Phi([\alpha]),\\
\Phi([\alpha])\cdot\Phi([\beta])-\Phi([\beta])\cdot\Phi([\alpha]) & =i\tau([\alpha],[\beta])\II,
\end{align}
for all $a,b\in\CC$ and for all $[\alpha],[\beta]\in\ClObs$. 
As usual, the first relation expresses the linearity of the quantum field, 
the second relation keeps track of the fact that classically the Proca field is a real field, 
therefore quantum Proca fields should be Hermitian, 
and finally the third relation implements the canonical commutation relations (CCR) for Bosonic field theories. 
We interpret $\Ac$ as the algebra of quantum observables for the Proca field 
over the globally hyperbolic spacetime $\Mb$. 

We conclude our investigations, analyzing certain properties of the quantum theory of the Proca field. 
Mimicking the proof of Theorem \ref{thmQuantumScalar} for the real scalar field, 
and exploiting the properties of the classical theory of the Proca field, 
which have been developed in Theorem \ref{thmClPropProca}, one obtains the following result. 

\begin{theorem}
Let  $\Mb=(\Mc,g,\ogth,\tgth)$ be an $n$-dimensional globally hyperbolic spacetime 
and let $\Ac$ be the unital $\ast$-algebra of observables for the Proca field introduced above. 
The following properties hold: 
\begin{description}
\item[{\bf Causality}] Elements of the algebra $\Ac$ localized in causally disjoint regions commute. 
More precisely, let $\alpha,\beta\in\fc^1(\Mc)$ be such that $\supp\alpha\cap J_\Mb(\supp\beta)=\emptyset$. 
Then $\Phi([\alpha])\cdot\Phi([\beta])=\Phi([\beta])\cdot\Phi([\alpha])$. 
\item[{\bf Time-slice axiom}] Let $\Oc\subset\Mc$ be a globally hyperbolic open neighborhood 
of a spacelike Cauchy surface $\Sigma$ for $\Mb$, namely $\Oc$ is an open neighborhood of $\Sigma$ in $\Mc$ 
containing all causal curves for $\Mb$ whose endpoints lie in $\Oc$. 
In particular, the restriction of $\Mb$ to $\Oc$ provides 
a globally hyperbolic spacetime $\Ob=(\Oc,g\vert_\Oc,\ogth\vert_\Oc,\tgth\vert_\Oc)$. 
Denote with $\Ac_\Mb$ and with $\Ac_\Ob$ the unital $\ast$-algebras of observables 
for the Proca field respectively over $\Mb$ and over $\Ob$. Then the unit-preserving $\ast$-homomorphism 
$\Phi(L):\Ac_\Ob\to\Ac_\Mb$, $\Phi([\alpha])\mapsto\Phi(L[\alpha])$ is an isomorphism of $\ast$-algebras, 
where $L$ is the symplectic isomorphism introduced in Theorem \ref{thmClPropProca}. 
\end{description}
\end{theorem}

\noindent To conclude the paper we would like to comment briefly on two aspects which have not been discussed. On the one hand we have only treated fields of spin $0$, $1/2$ and $1$, the latter under the assumption of a non-vanishing mass. This choice was made only for the sake of simplicity since all other cases would involve necessarily a discussion of local gauge invariance, a topic which is still under study and which would require a paper on its own -- for linear gauge theories refer to \cite{Hack:2012dm}. We mention a few references for an interested reader: for electromagnetism \cite{BDHS14, Benini:2013tra, Dappiaggi:2011zs, Dimock:1992ff, Fewster:2003ey, Pfenning:2009nx, Sanders:2012sf}, for spin $3/2$ fields \cite{Hack:2011yv, Hack:2012dm}, while for massless spin $2$ fields and linearized gravity \cite{BDM14, Fewster:2012bj}. Another important aspect, neglected in this paper, concerns the discussion about the existence of a relation between the algebra of observables for a free field theory built on two globally hyperbolic spacetimes which can be related one to the other via an isometric embedding. The analysis of such aspect leads to the formulation of the so-called principle of general local covariance, one of the milestones of modern axiomatic quantum field theory. This principle, together with its consequences, is discussed for example in \cite{FV}.

\section*{Acknowledgments}  
The work of M.B. has been supported partly by a grant of the University of Pavia and partly by a grant from the della Riccia foundation. Both are gratefully acknowledged. C.D. is grateful to Zhirayr Avetisyan for the useful discussions.

%%%%%%%%%%%%%%%%%%%%%%%% referenc.tex %%%%%%%%%%%%%%%%%%%%%%%%%%%%%%
% sample references
% %
% Use this file as a template for your own input.
%
%%%%%%%%%%%%%%%%%%%%%%%% Springer-Verlag %%%%%%%%%%%%%%%%%%%%%%%%%%
%
% BibTeX users please use
% \bibliographystyle{}
% \bibliography{}
%

\end{document}